\documentclass{sig-alternate}

\usepackage{comment}
\usepackage{multirow}
\usepackage[ruled,linesnumbered]{algorithm2e}
\usepackage[caption=false,font=footnotesize]{subfig}
\usepackage{ textcomp }
\usepackage{ amssymb }

\begin{document}
\title{Maiter: An Asynchronous Graph Processing Framework for Delta-based Accumulative Iterative Computation}

\numberofauthors{1}
\author{\alignauthor
Yanfeng Zhang$^{\dag\ddag}$, Qixin Gao$^{\dag}$, Lixin Gao$^{\ddag}$, Cuirong Wang$^{\dag}$\\
        \affaddr{\quad}\\
       \affaddr{$^{\dag}$Northeastern University, China}\\
       \affaddr{$^{\ddag}$University of Massachusetts Amherst}\\
       \email{\{yanfengzhang, lgao\}@ecs.umass.edu; \{gaoqx, wangcr\}@neuq.edu.cn}
}

\maketitle

\begin{abstract}
%\boldmath
Myriad of graph-based algorithms in machine learning and data mining require parsing relational data iteratively. These algorithms are implemented in a large-scale distributed environment in order to scale to massive data sets. To accelerate these large-scale graph-based iterative computations, we propose \emph{delta-based accumulative iterative computation} (DAIC). Different from traditional iterative computations, which iteratively update the result based on the result from the previous iteration, DAIC updates the result by accumulating the ``changes'' between iterations. By DAIC, we can process only the ``changes'' to avoid the negligible updates. Furthermore, we can perform DAIC asynchronously to bypass the high-cost synchronous barriers in heterogeneous distributed environments. Based on the DAIC model, we design and implement an asynchronous graph processing framework, Maiter. We evaluate Maiter on local cluster as well as on Amazon EC2 Cloud. The results show that Maiter achieves as much as 60x speedup over Hadoop and outperforms other state-of-the-art frameworks.
\end{abstract}

%%%%%%%%%%%%%%%%%%%%%%%%%%%%%%%%%%%%%%%%%%%%%%%%%%%%%%%%%%%%%%%% Main Text %%%%%%%%%%%%%%%%%%%%%%%%%%%%%%%%%%%%%%%%%%%%%%%%%%%%%%%%%%%%%%%%%%%

\section{Introduction}
\label{sec:intro}

The advances in data acquisition, storage, and networking technology have created huge collections of high-volume, high-dimensional relational data. Huge amounts of the relational data, such as Facebook user activities, Flickr photos, Web pages, and Amazon co-purchase records, have been collected. Making sense of these relational data is critical for companies and organizations to make better business decisions and even bring convenience to our daily life. Recent advances in
data mining, machine learning, and data analytics have led to a flurry of graph analytic techniques that typically require an iterative refinement process~\cite{1367618,1644932,Liben-Nowell:2007:LPS:1241540.1241551,Brin:1998:ALH:297810.297827}. However, the massive amount of data involved and potentially numerous iterations required make performing data analytics in a timely manner challenging. To address this challenge, MapReduce \cite{1251264,hadoop}, Pregel~\cite{1807184}, and a series of distributed frameworks \cite{distgraphlab,naiad,Zaharia:2010:SCC:1863103.1863113,1807184,piccolo} have been proposed to perform large-scale graph processing in a cloud environment.

Many of the proposed frameworks exploit vertex-centric programming model. Basically, the graph algorithm is described from a single vertex's perspective and then applied to each vertex for a loosely coupled execution. Given the input graph $G(V,E)$, each vertex $j\in V$ maintains a vertex state $v_{j}$, which is updated iteratively based on its in-neighbors' state, according to the update function $f$:
\begin{equation}
    v_{j}^{k}=f(v_{1}^{k-1}, v_{2}^{k-1}, \ldots, v_{n}^{k-1}),
\label{eq:concurrent}
\end{equation}
where $v_{j}^{k}$ represents vertex $j$'s state after the $k$ iterations, and $v_{1}, v_{2}, \ldots, v_{n}$ are the states of vertex $j$'s in-neighbors. The iterative process continues until the states of all vertices become stable, when the iterative algorithm converges.

Based on the vertex-centric model, most of the proposed frameworks leverage \textbf{synchronous iteration}. That is, the vertices perform the update in lock steps. At step $k$, vertex $j$ first collects $v_i^{k-1}$ from all its in-neighbors, followed by performing the update function $f$ to obtain $v_j^{k}$ based on these $v_i^{k-1}$. The synchronous iteration requires that all the update operations in the $(k-1)^{\text{th}}$ iteration have to complete before any of the update operations in the $k^{\text{th}}$ iteration start. Clearly, this synchronization is required in each step. These synchronizations might degrade performance, especially in heterogeneous distributed environments.

To avoid the high-cost synchronization barriers, \textbf{asynchronous iteration} was proposed \cite{Chazan1969199}. Performing updates asynchronously means that vertex $j$ performs the update at any time based on the most recent states of its in-neighbors. Asynchronous iteration has been studied in \cite{Chazan1969199,Baudet:1978:AIM:322063.322067,Bertsekas83distributedasynchronous}. Bypassing the synchronization barriers and exploiting the most recent state intuitively lead to more efficient iteration. However, asynchronous iteration might require more communications and perform useless computations. An activated vertex pulls all its in-neighbors' values, but not all of them have been updated, or even worse none of them is updated. In that case, asynchronous iteration performs a useless computation, which impacts efficiency. Furthermore, some asynchronous iteration cannot guarantee to converge to the same fixed point as synchronous iteration, which leads to uncertainty. 

In this paper, we propose \textbf{DAIC}, \emph{delta-based accumulative iterative computation}. In traditional iterative computation, each vertex state is updated based on its in-neighbors' previous iteration states. While in DAIC, each vertex propagates only the ``change'' of the state, which can avoid useless updates. The key benefit of only propagating the ``change'' is that, the ``changes'' can be accumulated monotonically and the iterative computation can be performed asynchronously. In addition, since the amount of ``change'' implicates the importance of an update, we can utilize more efficient priority scheduling for the asynchronous updates. Therefore, DAIC can be executed efficiently and asynchronously. Moreover, DAIC can guarantee to converge to the \textbf{same} fixed point. Given a graph iterative algorithm, we provide the sufficient conditions of rewriting it as a DAIC algorithm and list the guidelines on writing DAIC algorithms. We also show that a large number of well-known algorithms satisfy these conditions and illustrate their DAIC forms.

Based on the DAIC model, we design a distributed framework, \textbf{Maiter}. Maiter relies on Message Passing Interface (MPI) for communication and provides intuitive API for users to implement a DAIC algorithm. We systematically evaluate Maiter on local cluster as well as on Amazon EC2 Cloud~\cite{amazonec2}. Our results are presented in the context of four popular applications. The results show that Maiter can accelerate the iterative computations significantly. For example, Maiter achieves as much as 60x speedup over Hadoop for the well-known PageRank algorithm.

\section{Iterative Graph Processing}

The graph algorithm can be abstracted as the operations on a graph $G(V,E)$. Peoples usually exploit a vertex-centric model to solve the graph algorithms. Basically, the graph algorithm is described from a single vertex's perspective and then applied to each vertex for a loosely coupled execution. Iterative graph algorithms perform the same operations on the graph vertices for several iterations. Each vertex $j\in V$ maintains a vertex state $v_{j}$ that is updated iteratively. The key of a vertex-centric graph computation is the update function $f$ performed on each vertex $j$:
\begin{equation}
    v_{j}^{k}=f(v_{1}^{k-1}, v_{2}^{k-1}, \ldots, v_{n}^{k-1}),
\label{eq:concurrent}
\end{equation}
where $v_{j}^{k}$ represents vertex $j$'s state after the $k^{\text{th}}$ iteration, and $v_{1}, v_{2}, \ldots, v_{n}$ are the states of vertex $j$'s neighbors. The state values are passed between vertices through the edges. The iterative process continues until the graph vertex state becomes stable, when the iterative algorithm converges.

For example, the well-known PageRank algorithm iteratively updates all pages' ranking scores. According to the vertex-centric graph processing model, in each iteration, we update the ranking score of each page $j$, $R_j$, as follows:
\begin{equation}
    R_j^k=d\cdot \sum_{\{i|(i\rightarrow j)\in E\}}\frac{R_i^{k-1}}{|N(i)|}+(1-d),
\label{eq:pagerank}
\end{equation}
where $d$ is a damping factor, $|N(i)|$ is the number of outbound links of page $i$, $(i\rightarrow j)$ is a link from page $i$ to page $j$, and $E$ is the set of directed links. The PageRank scores of all pages are updated round by round until convergence.

In distributed computing, vertices are distributed to multiple processors and perform updates in parallel. For simplicity of exposition, assume that there are enough processors and each processor $j$ performs update for vertex $j$. All vertices perform the update in lock steps. At step $k$, vertex $j$ first collects $v_i^{k-1}$ from all its neighbor vertices, followed by performing the update function $f$ to obtain $v_j^{k}$ based on these $v_i^{k-1}$. The \textbf{synchronous iteration} requires that all the update operations in the $(k-1)^{\text{th}}$ iteration have to be completed before any of the update operations in the $k^{\text{th}}$ iteration starts. Clearly, this synchronization is required in each step. These synchronizations might degrade performance, especially in heterogeneous distributed environments.

To avoid the synchronization barriers, \textbf{asynchronous iteration} was proposed \cite{Chazan1969199}. Performing update operations asynchronously means that vertex $j$ performs the update
\begin{equation}
    v_{j}=f(v_{1}, v_{2}, \ldots, v_{n})
\label{eq:async}
\end{equation}
at any time based on the most recent values of its neighbor vertices, $\{v_{1}, v_{2}, \ldots, v_{n}\}$. The conditions of convergence of asynchronous iterations have been studied in \cite{Chazan1969199,Baudet:1978:AIM:322063.322067,Bertsekas83distributedasynchronous}.

By asynchronous iteration, as vertex $j$ is activated to perform an update, it ``pulls'' the values of its neighbor vertices, \emph{i.e.}, $\{v_1,v_2,\ldots,v_n\}$, and uses these values to perform an update on $v_j$. This scheme does not require any synchronization. However, asynchronous iteration intuitively requires more communications and useless computations than synchronous iteration. An activated vertex needs to pull the values from all its neighbor vertices, but not all of them have been updated, or even worse none of them is updated. In that case, asynchronous iteration performs a useless computation and results in significant communication overhead. Accordingly, ``pull-based'' asynchronous iteration is only applicable in an environment where the communication overhead is negligible, such as shared memory systems. In a distributed environment or in a cloud, ``pull-based'' asynchronous model cannot be efficiently utilized.

As an alternative, after vertex $i$ updates $v_i$, it ``pushes'' $v_i$ to all its neighbors $j$, and $v_i$ is buffered as $B_{i,j}$ on each vertex $j$, which will be updated as new $v_i$ arrives. Vertex $j$ only performs update when there are new values in the buffers and uses these buffered values $B_{i,j}$, to update $v_j$. In this way, the redundant communications can be avoided. However, the ``push-based'' asynchronous iteration results in higher space complexity. Each vertex $j$ has to buffer $|N(j)|$ values, where $|N(j)|$ is the number of vertex $j$'s neighbors. The large number of buffers also leads to considerable maintenance overhead.

To sum up, in a distributed environment, the synchronous iteration results in low performance due to the multiple global barriers, while the asynchronous iteration cannot be efficiently utilized due to various implementation overheads. Also note that, for some iterative algorithms, the asynchronous iteration cannot guarantee to converge to the same fixpoint as the synchronous iteration does, which leads to uncertainty.

\section{Delta-based Accumulative Iterative Computation (DAIC)}
\label{sec:daic}

In this section, we present delta-based accumulative iterative computation, DAIC. By DAIC, the graph iterative algorithms can be executed asynchronously and efficiently. We first introduce DAIC and point out the sufficient conditions of performing DAIC. Then, we propose DAIC's asynchronous execution model. We further prove its convergence and analyze its effectiveness. Under the asynchronous model, we also propose several scheduling policies to schedule the asynchronous updates.

\subsection{DAIC Introduction}
\label{sec:accum}

Based on the idea introduced in Section \ref{sec:intro}, we give the following 2-step update function of DAIC:
\begin{equation}
\left\{
\begin{aligned}
    v_j^k&=v_j^{k-1}\oplus \Delta v_j^{k},\\
    \Delta v_j^{k+1}&=\sum_{i=1}^{n}{\oplus g_{\{i,j\}}(\Delta v_i^{k})}.
\end{aligned}
\right.
\label{eq:sync}
\end{equation}
$k=1,2,\ldots$ is the iteration number. $v_j^k$ is the state of vertex $j$ after $k$ iterations. $\Delta v_j^k$ denotes the change from $v_j^{k-1}$ to $v_j^{k}$ in the `$\oplus$' operation manner, where `$\oplus$' is an abstract operator. $\displaystyle \sum_{i=1}^{n}{\oplus x_i}=x_1\oplus x_2\oplus\ldots\oplus x_n$ represents the accumulation of the ``changes'', where the accumulation is in the `$\oplus$' operation manner.

The first update function says that a vertex state $v_j^k$ is updated from $v_j^{k-1}$ by accumulating the change $\Delta v_j^k$. The second update function says that the change $\Delta v_j^{k+1}$, which will be used in the next iteration, is the accumulation of the received values $g_{\{i,j\}}(\Delta v_i^{k})$ from $j$'s various in-neighbors $i$. The propagated value from $i$ to $j$, $g_{\{i,j\}}(\Delta v_i^{k})$), is generated in terms of vertex $i$'s state change $\Delta v_i^{k}$. Note that, all the accumulative operation is in the `$\oplus$' operation manner.

However, not all iterative computation can be converted to the DAIC form. To write a DAIC, the update function should satisfy the following sufficient conditions.

The \textbf{first condition} is that,
\begin{itemize}
 \item update function $v_j^k=f(v_{1}^{k-1}, v_{2}^{k-1}, \ldots, v_{n}^{k-1})$ can be written in the form:
\end{itemize}
\begin{equation}
    v_j^k=g_{\{1,j\}}(v_{1}^{k-1})\oplus g_{\{2,j\}}(v_{2}^{k-1})\oplus \ldots \oplus g_{\{n,j\}}(v_{n}^{k-1})\oplus c_j
\label{eq:decompose}
\end{equation}
where $g_{\{i,j\}}(v_i)$ is a function applied on vertex $j$'s in-neighbor $i$, which denotes the value passed from vertex $i$ to vertex $j$. In other words, vertex $i$ passes value $g_{\{i,j\}}(v_{i})$ (instead of $v_i$) to vertex $j$. On vertex $j$, these $g_{\{i,j\}}(v_{i})$ from various vertices $i$ and $c_j$ are aggregated (by `$\oplus$' operation) to update $v_j$.

For example, the well-known PageRank algorithm satisfies this condition. It iteratively updates the PageRank scores of all pages. In each iteration, the ranking score of page $j$, $R_j$, is updated as follows:
\begin{equation}
    R_j^k=d\cdot \sum_{\{i|(i\rightarrow j)\in E\}}\frac{R_i^{k-1}}{|N(i)|}+(1-d),\notag
\label{eq:pagerank}
\end{equation}
where $d$ is a damping factor, $|N(i)|$ is the number of outbound links of page $i$, $(i\rightarrow j)$ is a link from page $i$ to page $j$, and $E$ is the set of directed links. The update function of PageRank is in the form of Equation (\ref{eq:decompose}), where $c_j=1-d$, `$\oplus$' is `$+$', and for any page $i$ that has a link to page $j$, $g_{\{i,j\}}(v_{i}^{k-1})=d\cdot\frac{v_i^{k-1}}{|N(i)|}$.

Next, since $\Delta v_j^k$ is defined to denote the ``change'' from $v_j^{k-1}$ to $v_j^{k}$ in the `$\oplus$' operation manner. That is,
\begin{equation}
\label{eq:derive1}
    v_j^{k}=v_j^{k-1}\oplus \Delta v_j^k,
\end{equation}
In order to derive $\Delta v_j^k$ we pose the \textbf{second condition}:
\begin{itemize}
  \item function $g_{\{i,j\}}(x)$ should have the \emph{distributive property} over `$\oplus$', \emph{i.e.}, $g_{\{i,j\}}(x\oplus y)=g_{\{i,j\}}(x)\oplus g_{\{i,j\}}(y)$.
\end{itemize}
By replacing $v_i^{k-1}$ in Equation (\ref{eq:decompose}) with $v_i^{k-2}\oplus \Delta v_i^{k-1}$, we have
\begin{equation}
\label{eq:derive3}
\begin{aligned}
    v_j^k=&g_{\{1,j\}}(v_{1}^{k-2})\oplus g_{\{1,j\}}(\Delta v_{1}^{k-1})\oplus\ldots\oplus\\
    &g_{\{n,j\}}(v_{n}^{k-2})\oplus g_{\{n,j\}}(\Delta v_{n}^{k-1})\oplus c_j.
\end{aligned}
\end{equation}
Further, let us pose the \textbf{third condition}:
\begin{itemize}
  \item operator `$\oplus$' should have the \emph{commutative property}, \emph{i.e.}, $x\oplus y=y\oplus x$;
  \item operator `$\oplus$' should have the \emph{associative property}, \emph{i.e.}, $(x\oplus y)\oplus z=x\oplus (y\oplus z)$;
\end{itemize}
Then we can combine these $g_{\{i,j\}}(v_{i}^{k-2})$, $i=1,2,\ldots,n$, and $c_j$ in Equation (\ref{eq:derive3}) to obtain $v_j^{k-1}$. Considering Equation (\ref{eq:derive1}), the combination of the remaining $g_{\{i,j\}}(\Delta v_{i}^{k-1})$, $i=1,2,\ldots,n$ in Equation (\ref{eq:derive3}), which is $\sum_{i=1}^{n}{\oplus g_{\{i,j\}}(\Delta v_i^{k-1})}$, will result in $\Delta v_i^k$. Then, we have the 2-step DAIC as shown in (\ref{eq:sync}).

To initialize a DAIC, we should set the start values of $v_j^0$ and $\Delta v_j^{1}$. $v_j^0$ and $\Delta v_j^{1}$ can be initialized to be any value, but the initialization should satisfy $v_j^{0}\oplus \Delta v_j^{1}=v_j^{1}=g_{\{1,j\}}(v_{1}^{0})\oplus g_{\{2,j\}}(v_{2}^{0})\oplus \ldots \oplus g_{\{n,j\}}(v_{n}^{0})\oplus c_j$, which is the \textbf{fourth condition}.

The PageRank's update function as shown in Equation (\ref{eq:pagerank}) satisfies all the conditions. $g_{\{i,j\}}(v_{i}^{k-1})=d\cdot\frac{v_i^{k-1}}{|N(i)|}$ satisfies the second condition. `$\oplus$' is `$+$', which satisfies the third condition. In order to satisfy the fourth condition, $v_j^0$ can be initialized to 0, and $\Delta v_j^{1}$ can be initialized to $1-d$.

To sum up, DAIC can be described as follows. Vertex $j$ first updates $v_j^k$ by accumulating $\Delta v_j^k$ (by `$\oplus$' operation) and then updates $\Delta v_j^{k+1}$ with $\sum_{i=1}^{n}{\oplus g_{\{i,j\}}(\Delta v_i^{k})}$.  We refer to $\Delta v_j$ as the \emph{delta value} of vertex $j$ and $g_{\{i,j\}}(\Delta v_i^{k})$ as the \emph{delta message} sent from $i$ to $j$. $\sum_{i=1}^{n}{\oplus g_{\{i,j\}}(\Delta v_i^{k})}$ is the accumulation of the received delta messages on vertex $j$ since the $k^{\text{th}}$ update. Then, the delta value $\Delta v_j^{k+1}$ will be used for the $(k+1)^{\text{th}}$ update. Apparently, this still requires all vertices to start the update synchronously. That is, $\Delta v_j^{k+1}$ has to accumulate all the delta messages $g_{\{i,j\}}(\Delta v_i^{k})$ sent from $j$'s in-neighbors, at which time it is ready to be used in the $(k+1)^{\text{th}}$ iteration. Therefore, we refer to the 2-step iterative computation in (\ref{eq:sync}) as \textbf{synchronous DAIC}.

\subsection{Asynchronous DAIC}

DAIC can be performed asynchronously. That is, a vertex can start update at any time based on whatever it has already received. We can describe \textbf{asynchronous DAIC} as follows, each vertex $j$ performs:
\begin{equation}
\label{eq:async}
\small
\begin{aligned}
\texttt{receive:}&\left\{
\begin{aligned}
&\texttt{Whenever receiving $m_j$,}\\
&\quad\Delta\check v_j\leftarrow \Delta\check v_j\oplus m_j.
\end{aligned}
\right.
\\
\texttt{update:}&\left\{
\begin{aligned}
& \check v_j\leftarrow\check v_j\oplus\Delta\check v_j;\\
& \texttt{For any $h$, if $g_{\{j,h\}}(\Delta\check v_j)\neq$ } \textbf{0},\\
& \quad\texttt{send value } g_{\{j,h\}}(\Delta\check v_j) \texttt{ to $h$};\\
& \Delta\check v_j\leftarrow\textbf{0},
\end{aligned}
\right.
\end{aligned}
\end{equation}
where $m_j$ is the received delta message $g_{\{i,j\}}(\Delta\check v_i)$ sent from any in-neighbor $i$. The \emph{receive} operation accumulates the received delta message $m_j$ to $\Delta\check v_j$. $\Delta\check v_j$ accumulates the received delta messages between two consecutive update operations. The \emph{update} operation updates $\check v_j$ by accumulating $\Delta\check v_j$, sends the delta message $g_{\{j,h\}}(\Delta\check v_j)$ to any of $j$'s out-neighbors $h$, and resets $\Delta\check v_j$ to \textbf{0}. Here, operator `$\oplus$' should have the \emph{identity property} of abstract value \textbf{0}, \emph{i.e.}, $x\oplus \textbf{0}=x$, so that resetting $\Delta\check v_j$ to \textbf{0} guarantees that the received value is cleared. Additionally, to avoid useless communication, it is necessary to check that the sent delta message $g_{\{j,h\}}(\Delta\check v_j)\neq \textbf{0}$ before sending.

For example, in PageRank, each page $j$ has a buffer $\Delta R_j$ to accumulate the received delta PageRank scores. When page $j$ performs an update, $R_j$ is updated by accumulating $\Delta R_j$. Then, the delta message $d\frac{\Delta R_j}{|N(j)|}$ is sent to $j$'s linked pages, and $\Delta R_j$ is reset to 0.

In asynchronous DAIC, the two operations on a vertex, receive and update, are completely independent from those on other vertices. Any vertex is allowed to perform the operations at any time. There is no lock step to synchronize any operation between vertices.

\subsection{Convergence}

To study the convergence property, we first give the following definition of the convergence of asynchronous DAIC.

\newtheorem{definition}{Definition}
\begin{definition}
\label{def:conv}
Asynchronous DAIC as shown in (\ref{eq:async}) converges as long as that after each vertex has performed the receive and update operations an infinite number of times, $\check v_j^{\infty}$ converges to a fixed value $\check v_j^*$.
\end{definition}
Then, we have the following theorem to guarantee that asynchronous DAIC will converge to the same fixed point as synchronous DAIC. Further, since synchronous DAIC is derived from the traditional form of iterative computation, \emph{i.e.}, Equation (\ref{eq:concurrent}), the asynchronous DAIC will converge to the same fixed point as traditional iterative computation.

\newtheorem{theorem}{Theorem}
\begin{theorem}
\label{th:aysnc}
If $v_j$ in (\ref{eq:concurrent}) converges, $\check v_j$ in (\ref{eq:async}) converges. Further, they converge to the same value, \emph{i.e.}, $v_j^{\infty}=\check v_j^{\infty}=\check v_j^*$.
\end{theorem}

We explain the intuition behind Theorem \ref{th:aysnc} as follows. Consider the process of DAIC as information propagation in a graph. Vertex $i$ with an initial value $\Delta v_i^1$ propagates delta message $g_{\{i,j\}}(\Delta v_i^1)$ to its out-neighbor $j$, where $g_{\{i,j\}}(\Delta v_i^1)$ is accumulated to $v_j$ and a new delta message $g_{\{j,h\}}(g_{\{i,j\}}(\Delta v_i^1))$ is produced and propagated to any of $j$'s out-neighbors $h$. By synchronous DAIC, the delta messages propagated from all vertices should be received by all their neighbors before starting the next round propagation. That is, the delta messages originated from a vertex are propagated strictly hop by hop. In contrast, by asynchronous DAIC, whenever some delta messages arrive, a vertex accumulates them to $\check v_j$ and propagates the newly produced delta messages to its neighbors. No matter synchronously or asynchronously, the spread delta messages are never lost, and the delta messages originated from each vertex will be eventually spread along all paths. For a destination node, it will eventually collect the delta messages originated from all vertices along various propagating paths. All these delta messages are eventually received and contributed to any $v_j$. Therefore, synchronous DAIC and asynchronous DAIC will converge to the same result.

In the following, we will show four lemmas to support our proof of Theorem \ref{th:aysnc}. The first two lemmas show the formal representations of a vertex state by synchronous DAIC and by asynchronous DAIC, respectively. The third lemma shows that there is a time instance when the result by asynchronous DAIC is smaller than or equal to the result by synchronous DAIC. Correspondingly, we show another lemma that there is a time instance when the result by asynchronous DAIC is larger than or equal to the result by synchronous DAIC. Once these lemmas are proved, it is sufficient to establish Theorem \ref{th:aysnc}.

\newtheorem{lemma}{Lemma}

\begin{lemma}
\label{lepr1}
By synchronous DAIC, $v_j$ after $k$ iterations is:
\begin{equation}
\small
\label{eq:syncproof}
\begin{aligned}
v_j^{k}=v_j^0\oplus \Delta v_j^1\oplus\sum_{l=1}^{k}\oplus\Big(\prod_{\{i_0,\ldots,i_{l-1},j\}\in P(j,l)}\oplus g_{\{i,j\}}(\Delta v_i^1)\Big),
\end{aligned}
\end{equation}
where
\begin{equation}
\small
\label{eq:syncproof2}
\begin{aligned}
\prod_{\{i_0,\ldots,i_{l-1},j\}}\oplus g_{\{i,j\}}(\Delta v_{i}^1)=g_{\{i_{l-1},j\}}(\ldots g_{\{i_1,i_2\}}(g_{\{i_0,i_1\}}(\Delta v_{i_0}^1)))\notag
\end{aligned}
\end{equation}
and $P(j,l)$ is a set of $l$-hop paths to reach node $j$.
\end{lemma}

\begin{proof}
According to the update functions shown in Equation (2) in the TPDS manuscript, after $k$ iterations, we have
\begin{equation}
\small
\label{eq:syncafterk4}
\begin{aligned}
v_j^{k}=&v_j^0\oplus \Delta v_j^1\oplus\Big(\sum_{i_1=1}^{n}\oplus g_{\{i_1,j\}}(\Delta v_{i_1}^1)\Big)\oplus\\
&\Big(\sum_{i_1=1}^{n}\oplus g_{\{i_1,j\}}\big(\sum_{i_2=1}^{n}\oplus g_{\{i_2,i_1\}}(\Delta v_{i_2}^1)\big)\Big)\oplus\ldots\oplus\\
&\Big(\sum_{i_1=1}^{n}\oplus g_{\{i_1,j\}}\big(\sum_{i_2=1}^{n}\oplus g_{\{i_2,i_1\}}\big(\ldots \sum_{i_k=1}^{n}\oplus g_{\{i_k,i_{k-1}\}}(\Delta v_{i_k}^1)\big)\big)\Big).\notag
\end{aligned}
\end{equation}
The $l^{\text th}$ term of the right side this equation corresponds to the received values from the $(l+1)$-hop away neighbors. Therefore, we have the claimed equation.
\end{proof}

In order to describe asynchronous DAIC, we define a continuous time instance sequence $\{t_1,t_2,\ldots,t_k\}$. Correspondingly, we define $S=\{S_1,S_2,\ldots,S_k\}$ as the series of subsets of vertices, where $S_k$ is a subset of vertices, and the propagated values of all vertices in $S_k$ have been received by their direct neighbors during the interval between time $t_{k-1}$ and time $t_k$ . As a special case, synchronous updates result from a sequence $\{V,V,\ldots,V\}$, where $V$ is the set of all vertices.

\begin{lemma}
\label{lepr2}
By asynchronous DAIC, following an activation sequence $S$, $\check v_j$ at time $t_k$ is:
\begin{equation}
\small
\label{eq:asyncproof}
\begin{aligned}
\check v_j^{k}=v_j^0\oplus \Delta v_j^1\oplus\sum_{l=1}^{k}\oplus\Big(\prod_{\{i_0,\ldots,i_{l-1},j\}\in P'(j,l)}\oplus g_{\{i,j\}}(\Delta v_i^1)\Big)
\end{aligned}
\end{equation}
where $P'(j,l)$ is a set of $l$-hop paths that satisfy the following conditions. First, $i_0\in S_l$. Second, if $l>0$, $i_1,\ldots,i_{l-1}$ respectively belongs to the sequence $S$. That is, there is $0<m_1<m_2<\ldots<m_{l-1}<k$ such that $i_h\in S_{m_{l-h}}$.
\end{lemma}

\begin{proof}
We can derive $\check v_j^{k}$ from Equation (6) in the TPDS manuscript.
\end{proof}

\begin{lemma}
\label{lm:big}
For any sequence $S$ that each vertex performs the receive and update operations an infinite number of times, given any iteration number $k$, we can find a subset index $k'$ in $S$ such that $|v_j^{*}-\check v_j^{k'}|\geq |v_j^{*}-v_j^k|$ for any vertex $j$.
\end{lemma}
\begin{proof}
Based on Lemma \ref{lepr1}, we can see that, after $k$ iterations, each node receives the
values from its direct/indirect neighbors as far as $k$ hops away, and it receives the values originated from each direct/indirect neighbor
once for each path. In other words, each node $j$ propagates its own
initial value $\Delta v_j^1$ (first to itself) and receives the
values from its direct/indirect neighbors through a path once.

Based on Lemma \ref{lepr2}, we can see that, after time $t_k$, each node receives values
from its direct/indirect neighbors as far as $k$ hops away, and it
receives values originated from each direct/indirect neighbor
through a path at most once. At time period $[t_{k-1},t_k]$, a value is received
from a neighbor only if the neighbor is in $S_k$. If the neighbor
is not in $S_k$, the value is stored at the neighbor or is on the way to other nodes. The node
will eventually receive the value as long as every node performs receive and update an infinite number of times.

As a result, $\check v_j^{k}$ receives values through a subset of the paths from $j$'s direct/indirect incoming neighbors within $k$ hops. In contrast, $v_j^k$ receives values through all paths from $j$'s direct/indirect incoming neighbors within $k$ hops. $\check v_j^{k}$ receives less values than $v_j^k$. Correspondingly, $\check v_j^{k}$ is further to the converged point $v_j^{*}$ than $v_j^{k}$. Therefore, we can set $k'=k$ and have the claim.
\end{proof}

\begin{lemma}
\label{lm:small}
For any sequence $S$ that each vertex performs the receive and update operations an infinite number of times, given any iteration number $k$, we can find a subset index $k''$ in $S$ such that $|v_j^{*}-\check v_j^{k''}|\leq |v_j^{*}-v_j^k|$ for any vertex $j$.
\end{lemma}
\begin{proof}
From the proof of Lemma~\ref{lm:big}, we know that $v_j^k$ receives values from all paths from direct/indirect
neighbors of $j$ within $k$ hops away. In order to let $\check v_j^{k''}$ receives all those values, we have to make sure
that all paths from direct/indirect
neighbors of $j$ within $k$ hops away are activated and their values are received. Since in sequence $S$ each vertex performs the update an infinite number of times, we can always find $k''$ such that $\{S_1, S_2, \ldots, S_{k''}\}$ contains all paths from direct and indirect
neighbors of $j$ within $k$ hops away. Correspondingly, $\check v_j^{k''}$ can be nearer to the converged point $v_j^{*}$ than $v_j^{k}$, or at least equal to. Therefore, we have the claim.
\end{proof}

Based on Lemma \ref{lm:big} and Lemma \ref{lm:small}, we have Theorem \ref{th:aysnc}.

\subsection{Effectiveness}
\label{sec:analysis}

As illustrated above, $v_j$ and $\check v_j$ both converge to the same fixed point. By accumulating $\Delta v_j$ (or $\Delta\check v_j$), $v_j$ (or $\check v_j$) either monotonically increases or monotonically decreases to a fixed value $v_j^*=v_j^{\infty}=\check v_j^{\infty}$. In this section, we show that $\check v_j$ converges faster than $v_j$.

To simplify the analysis, we first assume that 1) only one update occurs at any time point; 2) the transmission delay can be ignored, \emph{i.e.}, the delta message sent from vertex $i$, $g_{\{i,j\}}(\Delta v_i)$ (or $g_{\{i,j\}}(\Delta\check v_i)$), is directly accumulated to $\Delta v_j$ (or $\Delta\check v_j$).

The workload can be seen as the number of performed updates. Let \emph{update sequence} represent the update order of the vertices. By synchronous DAIC, all the vertices have to perform the update once and only once before starting the next round of updates. Hence, the update sequence is composed of a series of \emph{subsequences}. The length of each subsequence is $|V|$, \emph{i.e.}, the number of vertices. Each vertex occurs in a subsequence once and only once. We call this particular update sequence as \emph{synchronous update sequence}. While in asynchronous DAIC, the update sequence can follow any update order. For comparison, we will use the same synchronous update sequence for asynchronous DAIC.

By DAIC, no matter synchronously and asynchronously, the propagated delta messages of an update on vertex $i$ in subsequence $k$, \emph{i.e.}, $g_{\{i,j\}}(\Delta v_i^k)$ (or $g_{\{i,j\}}(\Delta\check v_i)$), are directly accumulated to $\Delta v_j^{k+1}$ (or $\Delta\check v_j$), $j=1,2,\ldots,n$. By synchronous DAIC, $\Delta v_j^{k+1}$ cannot be accumulated to $v_j$ until the update of vertex $j$ in subsequence $k+1$. In contrast, by asynchronous DAIC, $\Delta\check v_j$ is accumulated to $\check v_j$ immediately whenever vertex $j$ is updated after the update of vertex $i$ in subsequence $k$. The update of vertex $j$ might occur in subsequence $k$ or in subsequence $k+1$. If the update of vertex $j$ occurs in subsequence $k$, $\check v_j$ will accumulate more delta messages than $v_j$ after $k$ subsequences, which means that $\check v_j$ is closer to $v_j^*$ than $v_j$. Otherwise, $\check v_j=v_j$. Therefore, we have Theorem \ref{th:faster}.

\begin{theorem}
\label{th:faster}
Based on the same update sequence, after $k$ subsequences, we have $\check v_j$ by asynchronous DAIC and $v_j$ by synchronous DAIC. $\check v_j$ is closer to the fixed point $v_j^*$ than $v_j$ is, \emph{i.e.}, $|v_j^*-\check v_j|\leq |v_j^*-v_j|$.
\end{theorem}
\begin{proof}
In a single machine, the update sequence for asynchronous DAIC is a special $S$, where only one vertex in $S_k$ for any $k$ and any vertex is appeared once and only once in $\{S_{(k-1)n+1},S_{(k-1)n+2},\ldots,S_{(k-1)n+n}\}$ for any $k$, where $n$ is the total number of vertices. Based on Lemma \ref{lepr2}, we have
\begin{equation}
\small
\label{eq:asyncproof}
\begin{aligned}
\check v_j^{kn}=v_j^0\oplus \Delta v_j^1\oplus\sum_{l=1}^{kn}\oplus\Big(\prod_{\{i_0,\ldots,i_{l-1},j\}\in P'(j,l)}\oplus g_{\{i,j\}}(\Delta v_i^1)\Big),
\end{aligned}
\end{equation}
The values sent from any $k$-hop-away neighbors of $j$ will be received during time period $[t_{(k-1)n},t_{kn}]$, \emph{i.e.}, the sent values from $\{S_{(k-1)n+1},S_{(k-1)n+2},\ldots,S_{(k-1)n+n}\}$ are received. Further, $\check v_j^{kn}$ receives more values from further hops away, as far as $kn$-hop-away neighbors. Therefore, $\check v_j^{kn}$ is nearer to the converged point $v_j^{*}$ than $v_j^{k}$, \emph{i.e.}, $|v_j^{*}-\check v_j^{kn}|\leq |v_j^{*}-v_j^{k}|$.
\end{proof}

\subsection{Scheduling Policies}
\label{sec:schedule}

By asynchronous DAIC, we should control the update order of the vertices, \emph{i.e.}, specifying the scheduling policies. In reality, a subset of vertices are assigned to a processor, and multiple processors are running in parallel. The processor can perform the update for the assigned vertices in a round-robin manner, which is referred to as \emph{round-robin scheduling}. Moreover, it is possible to schedule the update of these local vertices dynamically by identifying their importance, which is referred to as \emph{priority scheduling}. In \cite{priter}, we have found that selectively processing a subset of the vertices has the potential of accelerating iterative computation. Some of the vertices can play an important decisive role in determining the final converged outcome. Giving an update execution priority to these vertices can accelerate the convergence.

In order to show the progress of the iterative computation, we quantify the iteration progress with $L_1$ norm of $\check v$, \emph{i.e.}, $||\check v||_1=\sum_{i}{\check v_{i}}$. Asynchronous DAIC either monotonically increases or monotonically decreases $||\check v||_1$ to a fixed point $||v^*||_1$. According to (\ref{eq:async}), an update of vertex $j$, \emph{i.e.}, $\check v_j=\check v_j\oplus\Delta\check v_j$, either increases $||\check v||_1$ by $(\check v_j\oplus\Delta\check v_j-\check v_j)$ or decreases $||\check v||_1$ by $(\check v_j-\check v_j\oplus\Delta\check v_j)$. Therefore, by priority scheduling, vertex $j=arg\max_j{|\check v_j\oplus\Delta\check v_j-\check v_j|}$ is scheduled first. In other words, The bigger $|\check v_j\oplus\Delta\check v_j-\check v_j|$ is, the higher update priority vertex $j$ has. For example, in PageRank, we set each page $j$'s scheduling priority based on $|R_j+\Delta R_j-R_j|=\Delta R_j$. Then, we will schedule page $j$ with the largest $\Delta R_j$ first. To sum up, by priority scheduling, the vertex $j=arg\max_j{|\check v_j\oplus\Delta\check v_j-\check v_j|}$ is scheduled for update first.

To guarantee the convergence of asynchronous DAIC under the priority scheduling, we first pose the following lemma.

\begin{lemma}
\label{lm:infinite}
By asynchronous priority scheduling, $\check v_j'$ converges to the same fixed point $v_j^*$ as $v_j$ by synchronous iteration converges to, \emph{i.e.}, $\check v_j'^{\infty}=v_j^{\infty}=v_j^*$.
\end{lemma}

\begin{proof}
There are two cases to guide priority scheduling. We only prove the case that schedules vertex $j$ that results in the largest $(\check v_j\oplus\Delta \check v_j-\check v_j)$. The proof of the other case is similar.

We prove the lemma by contradiction. Assume there is a set of vertices, $S_*$, which is scheduled to perform update only before time $t_*$. Then the accumulated values on the vertices of $S_*$, $\check v_{S_*}$, will not change since then. While they might receive values from other vertices, \emph{i.e.}, $||\check v_{S_*}\oplus\Delta\check v_{S_*}-\check v_{S_*}||_1$ might become larger. On the other hand, the other vertices ($V-S_*$) continue to perform the update operation, the received values on them, $\Delta\check v_{V-S_*}$, are accumulated to $\check v_{V-S_*}$ and propagated to other vertices again. As long as the iteration converges, the difference between the results of two consecutive updates, $||\check v_{V-S_*}\oplus\Delta\check v_{V-S_*}- v_{V-S_*}||_1$ should decrease ``steadily'' to 0. Therefore, eventually at some point,
\begin{equation}
\label{eq:lemma:infinite1}
\frac{||\check v_{S_*}\oplus\Delta\check v_{S_*}-\check v_{S_*}||_1}{|S_*|}>||\check v_{V-S_*}\oplus\Delta\check v_{V-S_*}-\check v_{V-S_*}||_1.
\end{equation}
That is,
\begin{equation}
\label{eq:lemma:infinite1}
\max_{j\in S_*}(\check v_j\oplus\Delta\check v_j-\check v_j)>\max_{j\in V-S_*}(\check v_j\oplus\Delta\check v_j-\check v_j).
\end{equation}
Since the vertex that has the largest $(\check v_j\oplus\Delta \check v_j-\check v_j)$ should be scheduled under priority scheduling, a vertex in $S_*$ should be scheduled at this point, which contradicts with the assumption that any vertex in $S_*$ is not scheduled after time $t_*$.
\end{proof}

Then, with the support of Lemma \ref{lm:infinite} and Theorem \ref{th:aysnc}, we have Theorem \ref{th:pri}.

\begin{theorem}
\label{th:pri}
By asynchronous priority scheduling, $\check v_j'$ converges to the same fixed point $v_j^*$ as $v_j$ by synchronous iteration converges to, \emph{i.e.}, $\check v_j'^{\infty}=v_j^{\infty}=v_j^*$.
\end{theorem}

Furthermore, according to the analysis presented above, we have Theorem \ref{th:faster2} to support the effectiveness of priority scheduling.

\begin{theorem}
\label{th:faster2}
Based on asynchronous DAIC, after the same number of updates, we have $\check v_j'$ by priority scheduling and $\check v_j$ by round-robin scheduling. $\check v_j'$ is closer to the fixed point $v_j^*$ than $\check v_j$ is, \emph{i.e.}, $|v_j^*-\check v_j'|\leq |v_j^*-\check v_j|$.
\end{theorem}

\section{DAIC algorithms}

In this section, we provide the guidelines of writing DAIC algorithms and present a broad class of DAIC algorithms. 

\subsection{Writing DAIC algorithms}
\label{sec:app}

Given an iterative algorithm, the following steps are recommended for converting it to a DAIC algorithm.

\begin{itemize}
  \item \textbf{Step1: Vertex-Centric Check}. Check whether the update function is applied on each vertex, and write the vertex-centric update function $f$. If not, try to rewrite the update function.
  \item \textbf{Step2: Formation Check}. Check whether $f$ is in the form of Equation (\ref{eq:decompose})? If yes, identify the sender-based function $g_{\{i,j\}}(v_i)$ applied on each sender vertex $i$, the abstract operator `$\oplus$' for accumulating the received delta messages on receiver vertex $j$.
  \item \textbf{Step3: Properties Check}. Check whether $g_{\{i,j\}}(v_i)$ has the distributive property and whether operator `$\oplus$' has the commutative property and the associative property?
  \item \textbf{Step4: Initialization}. According to (\ref{eq:sync}), initialize $v_j^0$ and $\Delta v_j^{1}$ to satisfy $v_j^{1}=v_j^{0}\oplus \Delta v_j^{1}$, and write the iterative computation in the 2-step DAIC form.
  \item \textbf{Step5: Priority Assignment} (Optional). Specify the scheduling priority of each vertex $j$ as $|\check v_j\oplus\Delta\check v_j-\check v_j|$ for scheduling the asynchronous updates.
\end{itemize}

\subsection{DAIC algorithms}
\label{sec:app2}

\begin{table*}[!t]
\small
    \caption{A list of DAIC algorithms}
    \label{tab:algorithms}
    \centering
    \begin{tabular}{c|c|c|c|c}
    \hline
    algorithm & $g_{\{i,j\}}(x)$ & $\oplus$ & $v_j^0$ & $\Delta v_j^1$\\
\hline\hline
SSSP & $x+A(i,j)$ & $\min$ & $\infty$ & 0 ($j=s$) or $\infty$ ($j\neq s$)\\
\hline
Connected Components & $A(i,j)\cdot x$ & $\max$ & -1 & $j$\\
\hline
PageRank & $d\cdot A(i,j)\cdot\frac{x}{|N(j)|}$ & $+$ & 0 & $1-d$\\
\hline
Adsorption & $p_{i}^{cont}\cdot A(i,j)\cdot x$ & $+$ & 0 & $p_{j}^{inj}\cdot{I_{j}}$\\
\hline
HITS ($authority$) & $d\cdot A(i,j)\cdot x$ & $+$ & 0 & 1\\
\hline
Katz metric & $\beta\cdot A(i,j)\cdot x$ & $+$ & 0 & 1 ($j=s$) or 0 ($j\neq s$) \\
\hline
Jacobi method & $-\frac{A_{ji}}{A_{jj}}\cdot x$ & $+$ & 0 & $\frac{b_j}{A_{jj}}$ \\
\hline
SimRank & $\frac{C\cdot A(i,j)}{|I(a)||I(b)|}\cdot x$ & $+$ & $|I(a)\cap I(b)|$ ($a\neq b$) or 1 ($a=b$) & $\frac{|I(a)||I(b)|}{C}$ ($a\neq b$) or 0 ($a=b$)\\
\hline
%Expected Hitting Time & 1 ($i=s$) or 0 ($i\neq s$) & $P(j,i)\cdot x$ & SUM & $\downarrow \Delta v_i \cdot \sum_{j}P(j,i)$\\
%\hline
Rooted PageRank & $A(j,i)\cdot x$ & $+$ & 0 & 1 ($j=s$) or 0 ($j\neq s$) \\
\hline
\end{tabular}
\end{table*}

Following the guidelines, we have found a broad class of DAIC algorithms, including single source shortest path (SSSP), PageRank, linear equation solvers, Adsorption, SimRank, etc. Table~\ref{tab:algorithms} shows a list of such algorithms. Each of their update functions is represented with a tuple ($g_{\{i,j\}}(x)$, $\oplus$, $v_j^0$, $\Delta v_j^1$). In Table \ref{tab:algorithms}, matrix $A$ represents the graph adjacency information. If there is an edge from vertex $i$ to vertex $j$, $A_{i,j}$ represents the edge weight from $i$ to $j$, or else $A_{i,j}=0$.

\subsubsection{Single Source Shortest Path}
The \emph{single source shortest path} algorithm (SSSP) has been widely used in online social networks and web mapping. Given a source node $s$, the algorithm derives the shortest distance from $s$ to all the other nodes on a directed weighted graph. Initially, each node $j$'s distance $d_j^0$ is initialized to be $\infty$ except that the source $s$'s distance $d_s^0$ is initialized to be 0. In each iteration, the shortest distance from $s$ to $j$, $d_j$, is updated with the following update function:
\begin{equation}
\begin{aligned}
d_j^k=\min\{d_1^{k-1}+A(1,j), &d_2^{k-1}+A(2,j), \ldots, \\
&d_n^{k-1}+w(n,j), d_j^0\},\notag
\end{aligned}
\end{equation}
where $A(i,j)$ is the weight of an edge from node $i$ to node $j$, and $A(i,j)=\infty$ if there is no edge between $i$ and $j$. The update process is performed iteratively until convergence, where the distance values of all nodes no longer change.

Following the guidelines, we identify that operator `$\oplus$' is `$\min$', function $g_{\{i,j\}}(d_i) = d_i + A(i,j)$. Apparently, the function $g_{\{i,j\}}(x)$ has the distributive property, and the operator `$\min$' has the commutative and associative properties. The initialization can be $d_j^0=\infty$ and $\Delta d_j^1=0$ if $j=s$, or else $\Delta d_j=\infty$. Therefore, SSSP can be performed by DAIC. Further, suppose $\Delta d_j$ is used to accumulate the received distance values by `$\min$' operation, the scheduling priority of node $j$ would be $d_j-\min\{d_j,\Delta d_j\}$.

\subsubsection{Linear Equation Solvers}

Generally, DAIC can be used to solve systems of linear equations of the form
\begin{equation}
\label{eq:linear}
A\cdot\chi=b,\notag
\end{equation}
where $A$ is a sparse $n\times n$ matrix with each entry $A_{ij}$, and $\chi,b$ are size-$n$ vectors with each entry $\chi_j$, $b_j$ respectively.

One of the linear equation solvers, \emph{Jacobi method}, iterates each entry of $\chi$ as follows:
\begin{equation}
\label{eq:linear}
\chi_j^k=-\frac{1}{A_{jj}}\cdot\sum_{i\neq j}A_{ji}\cdot\chi_i^{k-1}+\frac{b_j}{A_{jj}}.\notag
\end{equation}
The method is guaranteed to converge if the spectral radius of the iteration matrix is less than 1. That is, for any matrix norm $||\cdot||$, $\lim_{k\rightarrow \infty}||B^k||^{\frac{1}{k}}<1$, where $B$ is the matrix with $B_{ij}=-\frac{A_{ij}}{A_{ii}}$ for $i\neq j$ and $B_{ij}=0$ for $i=j$.

Following the guidelines, we identify that operator `$\oplus$' is `$+$', function $g_{\{i,j\}}(\chi_i)=-\frac{A_{ji}}{A_{jj}}\cdot \chi_i$. Apparently, the function $g_{\{i,j\}}(x)$ has the distributive property, and the operator `$+$' has the commutative and associative properties. The initialization can be $\chi_j^0=0$ and $\Delta \chi_j^1=\frac{b_j}{A_{jj}}$. Therefore, the Jacobi method can be performed by DAIC. Further, suppose $\Delta \chi_j$ is used to accumulate the received delta message, the scheduling priority of node $j$ would be $\Delta \chi_j$.

\subsubsection{PageRank}

The \emph{PageRank} algorithm \cite{Brin:1998:ALH:297810.297827} is a popular algorithm proposed for ranking web pages. Initially, the PageRank scores are evenly distributed among all pages. In each iteration, the ranking score of page $j$, $R_j$, is updated as follows:
\begin{equation}
\label{eq:pr}
R_j=d\cdot \sum_{\{i|(i\rightarrow j)\in E\}}\frac{R_i}{|N(i)|}+(1-d),
\end{equation}
where $d$ is damping factor, $|N(i)|$ is the number of outbound links of page $i$, and $E$ is the set of link edges. The iterative process terminates when the sum of changes of two consecutive iterations is sufficiently small. The initial guess of $R_i$ can be any value. In fact, the final converged ranking score is independent from the initial value.

Following the guidelines, we identify that operator `$\oplus$' is `$+$', function $g_{\{i,j\}}(R_i)=d\cdot A_{i,j}\frac{R_i}{N(i)}$, where $A$ represents the adjacency matrix and $A_{i,j}=1$ if there is a link from $i$ to $j$ or else $A_{i,j}=0$. Apparently, the function $g_{\{i,j\}}(x)$ function has distributive property and the operator `$+$' has the commutative and associative properties. The initialization can be $R_j^0=0$ and $\Delta R_j^1=1-d$. Therefore, PageRank can be performed by DAIC. Further, suppose $\Delta R_j$ is used to accumulate the received PageRank values, the scheduling priority of node $j$ would be $\Delta R_j$.

\subsubsection{Adsorption}

\emph{Adsorption} \cite{1367618} is a graph-based label propagation algorithm that provides personalized recommendation for contents (e.g., video, music, document, product). The concept of \emph{label} indicates a certain common feature of the entities. Given a weighted graph $G=(V, E)$, where $V$ is the set of nodes, $E$ is the set of edges. $A$ is a column normalized matrix (i.e., $\sum_{i}A(i,j)=1$) indicating that the sum of a node's inbound links' weight is equal to 1. Node $j$ carries a probability distribution $L_j$ on label set $L$, and each node $j$ is initially assigned with an \emph{initial distribution} $I_j$. The algorithm proceeds as follows. For each node $j$, it iteratively computes the weighted average of the label distributions from its neighboring nodes, and then uses the random walk probabilities to estimate a new label distribution as follows:
\begin{equation}
\label{eq:adsorption}
L_j^k=p_{j}^{cont}\cdot{\sum_{\{i|(i\rightarrow j)\in E\}}A(i,j)\cdot{L_i^{k-1}}+p_{j}^{inj}\cdot{I_{j}}},\notag
\end{equation}
where $p_{j}^{cont}$ and $p_{j}^{inj}$ are constants associated with node $j$. If Adsorption converges, it will converge to a unique set of label distributions.

%The iterative process terminates when the sum of changes between two consecutive iterations is sufficiently small.

Following the guidelines, we identify that operator `$\oplus$' is `$+$', $g_{\{i,j\}}(L_i)=p_{j}^{cont}\cdot A(i,j)\cdot L_i$. Apparently, the function $g_{\{i,j\}}(x)$ has the distributive property, and the operator `$+$' has the commutative and associative properties. The initialization can be $L_j^0=0$ and $\Delta L_j^1=p_{j}^{inj}\cdot{I_{j}}$. Therefore, Adsorption can be performed by accumulative updates. Further, suppose $\Delta L_j$ is used to accumulate the received distance values, the scheduling priority of node $j$ would be $\Delta L_j$.

\subsubsection{SimRank}

SimRank \cite{Jeh:2002:SMS:775047.775126} was proposed to measure the similarity between two nodes in the network. It has been successfully used for many applications in social networks, information retrieval, and link prediction. In SimRank, the similarity between two nodes (or objects) $a$ and $b$ is defined as the average similarity between nodes linked with $a$ and those with $b$. Mathematically, we iteratively update $s(a,b)$ as the similarity value between node $a$ and $b$:
\begin{equation}
s^k(a,b)=\frac{C}{|I(a)||I(b)|}\sum_{c\in I(a),d\in I(b)}s^{k-1}(c,d),\notag
\label{eq:simrank}
\end{equation}
where $s^1(a,b)=1$ if $a=b$, or else $s^1(a,b)=0$, $I(a)={b\in V |(b,a)\in E}$ denoting all the nodes that have a link to $a$, and $C$ is a decay factor satisfying $0<C<1$.

However, this update function is applied on node-pairs. It is not a vertex-centric update function. We should rewrite the update function. Cao \emph{et. al.} has proposed \emph{Delta-SimRank} \cite{delta-simrank}. They first construct a node-pair graph $G^2=\{V^2,E^2\}$. Each node denotes one pair of nodes of the original graph. One node $ab$ in $G^2$ corresponds to a pair of nodes $a$ and $b$ in $G$. There is one edge $(ab,cd)\in E^2$ if $(a,c)\in E$ and $(b,d)\in E$. If the graph size $|G|=n$, the node-pair graph size $|G^2|=n^2$. Let vertex $j$ represent $ab$ and vertex $i$ represent $cd$. Then, the update function of a vertex $j\in G^2$ is:
\begin{equation}
s^k(j)=\frac{C}{|I(a)||I(b)|}\sum_{i\in I(j)}s^{k-1}(i),\notag
\label{eq:deltasimrank}
\end{equation}
where $I(a)$ and $I(b)$ denote the neighboring nodes of $a$ and $b$ in $G$ respectively, and $I(j)$ denotes the neighboring nodes of $j$ in $G^2$.

The new form of SimRank update function in the node-pair graph $G^2$ is vertex-centric. Following the DAIC guidelines, we identify that operator `$\oplus$' is `$+$', and function $g_{\{i,j\}}(s(i))=\frac{C\cdot A(i,j)}{|I(a)||I(b)|}\cdot s(i)$, where $A_{i,j}=1$ if $i\in I(j)$ (\emph{i.e.}, $cd\in I(ab)$) or else $A_{i,j}=0$. Apparently, the function $g_{\{i,j\}}(x)$ has the distributive property, and the operator `$+$' has the commutative and associative properties. The initialization of $s^0(j)$ can be $s^0(j)=s^0(ab)=1$ if $a=b$, or else $s^0(j)=s^0(ab)=\sum_{c\in I(a)\&c\in I(b)}1=|I(a)\cap I(b)|$. The initialization of $\Delta s^1(j)$ can be $\Delta s^1(j)=\Delta s^1(ab)=0$ if $a=b$, or else $\Delta s^1(j)=\Delta s^1(ab)=\frac{|I(a)||I(b)|}{C}$. Therefore, SimRank can be performed by DAIC. Further, suppose $\Delta s(j)$ is used to accumulate the received delta messages, the scheduling priority of node $j$ would be $\Delta s(j)$.

\subsubsection{Other Algorithms}

We have shown several typical DAIC algorithms. Following the guidelines, we can rewrite them in DAIC form. In addition, there are many other DAIC algorithms. Table 1 of the main manuscript shows a list of such algorithms. Each of their update functions is represented with a tuple ($g_{\{i,j\}}(x)$, $\oplus$, $v_j^0$, $\Delta v_j^1$).

The \emph{Connected Components} algorithm~\cite{5360248} finds connected components in a graph (the graph adjacency information is represented in matrix $A$, $A_{i,j}=1$ if there is a link from $i$ to $j$ or else $A_{i,j}=0$). Each node updates its component id with the largest received id and propagates its component id to its neighbors, so that the algorithm converges when all the nodes belonging to the same connected component have the same component id.

\emph{Hyperlink-Induced Topic Search} (HITS)~\cite{Kleinberg:1999:ASH:324133.324140} ranks web pages in a web linkage graph $W$ by a 2-phase iterative update, the \emph{authority update} and the \emph{hub update}. Similar to Adsorption, the authority update requires each node $i$ to generate the output values damped by $d$ and scaled by $A(i,j)$, where matrix $A=W^TW$, while the hub update scales a node's output values by $A'(i,j)$, where matrix $A'=WW^T$.

The \emph{Katz metric}~\cite{katz1953} is a proximity measure between two nodes in a graph (the graph adjacency information is represented in matrix $A$, $A_{i,j}=1$ if there is a link from $i$ to $j$, or else $A_{i,j}=0$). It is computed as the sum over the collection of paths between two nodes, exponentially damped by the path length with a damping factor $\beta$.

\emph{Rooted PageRank}~\cite{1644932} captures the probability for any node $j$ running into node $s$, based on the node-to-node proximity, $A(j,i)$, indicating the probability of jumping from node $j$ to node $i$.

\section{Maiter}
\label{sec:design}

\subsection{System Design}
To support implementing a DAIC algorithm in a large-scale distributed manner and in a highly efficient asynchronous manner, we propose an asynchronous distributed framework, Maiter. Users only need to follow the guidelines to specify the function $g_{\{i,j\}}(v_i)$, the abstract operator `$\oplus$', and the initial values $v_j^0$ and $\Delta v_j^1$ through Maiter API (Maiter API is described in Section 4 of the supplementary file). The framework will automatically deploy these DAIC algorithms in the distributed environment and perform asynchronous iteration efficiently.

Maiter is implemented by modifying Piccolo \cite{piccolo}, and Maiter's source code is available online \cite{maiter}. It relies on message passing for communication between vertices. In Maiter, there is a master and multiple workers. The master coordinates the workers and monitors the status of workers. The workers run in parallel and communicate with each other through MPI. Each worker performs the update for a subset of vertices. In the following, we introduce Maiter's key functionalities.

\textbf{Data Partition.} Each worker loads a subset of vertices in memory for processing. Each vertex is indexed by a global unique \emph{vid}. The assignment of a vertex to a worker depends solely on the vid. A vertex with vid $j$ is assigned to worker $h(j)$, where $h()$ is a hash function applied on the vid. Besides, preprocessing for smart graph partition can be useful. For example, one can use a lightweight clustering algorithm to preprocess the input graph, assigning the strongly connected vertices to the same worker, which can reduce communication.

\begin{figure}[!t]
  \centering
  \includegraphics[width=2.2in]{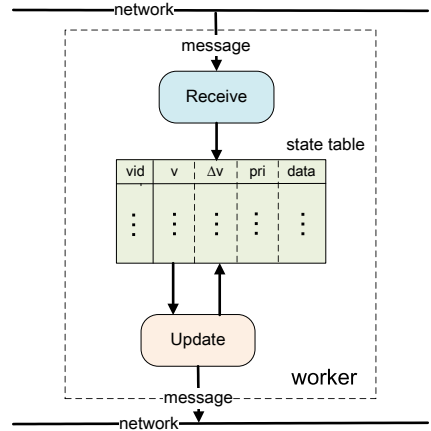}
  \caption{Worker overview.}
  \label{fig:iterupdate}
\end{figure}

\textbf{Local State Table.} The vertices in a worker are maintained in a local in-memory key-value store, \emph{state table}. Each state table entry corresponds to a vertex indexed by its vid. As depicted in Fig. ~\ref{fig:iterupdate}, each table entry contains five fields. The 1st field stores the vid $j$ of a vertex; the 2rd field stores $v_j$; the 3rd field stores $\Delta v_j$; the 4th field stores the priority value of vertex $j$ for priority scheduling; the 5th field stores the input data associated with vertex $j$, such as the adjacency list. Users are responsible for initializing the $v$ fields and the $\Delta v$ fields through the provided API. The priority fields are automatically initialized based on the values of the $v$ fields and $\Delta v$ fields. Users read an input partition and fills entry $j$'s data field with vertex $j$'s input data.

\textbf{Receive Thread and Update Thread.} As described in Equation (\ref{eq:async}), DAIC is accomplished by two key operations, the receive operation and the update operation. In each worker, these two operations are implemented in two threads, the \emph{receive thread} and the \emph{update thread}. The receive thread performs the receive operation for all local vertices. Each worker receives the delta messages from other workers and updates the $\Delta v$ fields by accumulating the received delta messages. The update thread performs the update operation for all local vertices. When operating on a vertex, it updates the corresponding entry's $v$ field and $\Delta v$ field, and sends messages to other vertices.

\textbf{Scheduling within Update Thread.} The simplest scheduling policy is to schedule the local vertices for update operation in a round robin fashion. The update thread performs the update operation on the table entries in the order that they are listed in the local state table and round-by-round. The static scheduling is simple and can prevent starvation.

However, as discussed in Section \ref{sec:schedule}, it is beneficial to provide priority scheduling. In addition to the static round-robin scheduling, Maiter supports dynamic priority scheduling. A \emph{priority queue} in each worker contains a subset of local vids that have larger priority values. The update thread dequeues the vid from the priority queue, in terms of which it can position the entry in the local state table and performs an update operation on the entry. Once all the vertices in the priority queue have been processed, the update thread extracts a new subset of high-priority vids for next round update. The extraction of vids is based on the priority field. Each entry's priority field is initially calculated based on its initial $v$ value and $\Delta v$ value. During the iterative computation, the priority field is updated whenever the $\Delta v$ field is changed (\emph{i.e.}, whenever some delta messages are received).

The number of extracted vids in each round, \emph{i.e.}, the priority queue size, balances the tradeoff between the gain from accurate priority scheduling and the cost of frequent queue extractions. The priority queue size is set as a portion of the state table size. For example, if the queue size is set as 1\% of the state table size, we will extract the top 1\% high priority entries for processing. In addition, we also use the sampling technique proposed in \cite{priter} for efficient queue extraction, which only needs $O(N)$ time, where $N$ is the number of entries in local state table.

\textbf{Message Passing.} Maiter uses OpenMPI \cite{openmpi} to implement message passing between workers. A message contains a vid indicating the message's destination vertex and a value. Suppose that a message's destination vid is $k$. The message will be sent to worker $h(k)$, where $h()$ is the partition function for data partition, so the message will be received by the worker where the destination vertex resides.

A naive implementation of message passing is to send the output messages as soon as they are produced. This will reach the asynchronous iteration's full potential. However, initializing message passing leads to system overhead. To reduce this overhead, Maiter buffers the output messages and flushes them after a timeout. If a message's destination worker is the host worker, the output message is directly applied to the local state table. Otherwise, the output messages are buffered in multiple \emph{msg table}s, each of which corresponds to a remote destination worker. We can leverage early aggregation on the msg table to reduce network communications. Each msg table entry consists of a destination vid field and a value field. As mentioned in Section \ref{sec:accum}, the associative property of operator `$\oplus$', \emph{i.e.}, $(x\oplus y)\oplus z=x\oplus (y\oplus z)$, indicates that multiple messages with the same destination can be aggregated at the sender side or at the receiver side. Therefore, by using the msg table, Maiter worker combines the output messages with the same vid by `$\oplus$' operation before sending them.

\textbf{Iteration Termination.} To terminate iteration, Maiter exploits \emph{progress estimator} in each worker and a global \emph{terminator} in the master. The master periodically broadcasts a \emph{progress request signal} to all workers. Upon receipt of the termination check signal, the progress estimator in each worker measures the iteration progress locally and reports it to the master. The users are responsible for specifying the progress estimator to retrieve the iteration progress by parsing the local state table. After the master receives the local iteration progress reports from all workers, the terminator makes a global termination decision in respect of the global iteration progress, which is calculated based on the received local progress reports. If the terminator determines to terminate the iteration, the master broadcasts a \emph{terminate signal} to all workers. Upon receipt of the terminate signal, each worker stops updating the state table and dumps the local table entries to HDFS, which contain the converged results. Note that, even though we exploit a synchronous termination check periodically, it will not impact the asynchronous computation. The workers proceed after producing the local progress reports without waiting for the master's feedback.

\textbf{Fault Tolerance.} The fault tolerance support for synchronous computation models can be performed through checkpointing, where the state data is checkpointed on the reliable HDFS every several iterations. If some workers fail, the computation rolls back to the most recent iteration checkpoint and resumes from that iteration. Maiter exploits Chandy-Lamport \cite{Chandy:1985:DSD:214451.214456} algorithm to design asynchronous iteration's fault tolerance mechanism. The checkpointing in Maiter is performed at regular time intervals rather than at iteration intervals. The state table in each worker is dumped to HDFS every period of time. However, during the asynchronous computation, the information in the state table might not be intact, in respect that the messages may be on their way to act on the state table. To avoid missing messages, not only the state table is dumped to HDFS, but also the msg tables in each worker are saved. Upon detecting any worker failure (through probing by the master), the master restores computation from the last checkpoint, migrates the failed worker's state table and msg tables to an available worker, and notifies all the workers to load the data from the most recent checkpoint to recover from worker failure. For detecting master failure, Maiter can rely on a secondary master, which restores the recent checkpointed state to recover from master failure.

\subsection{Maiter API}

Users implement a Maiter program using the provided API, which is written in C++ style. A DAIC algorithm is specified by implementing three functionality components, \texttt{Partitioner}, \texttt{IterateKernel}, and \texttt{TermChecker} as shown in Fig. \ref{fig:api}.

\begin{figure}[!t]
  \centering
  \includegraphics[width=3.3in]{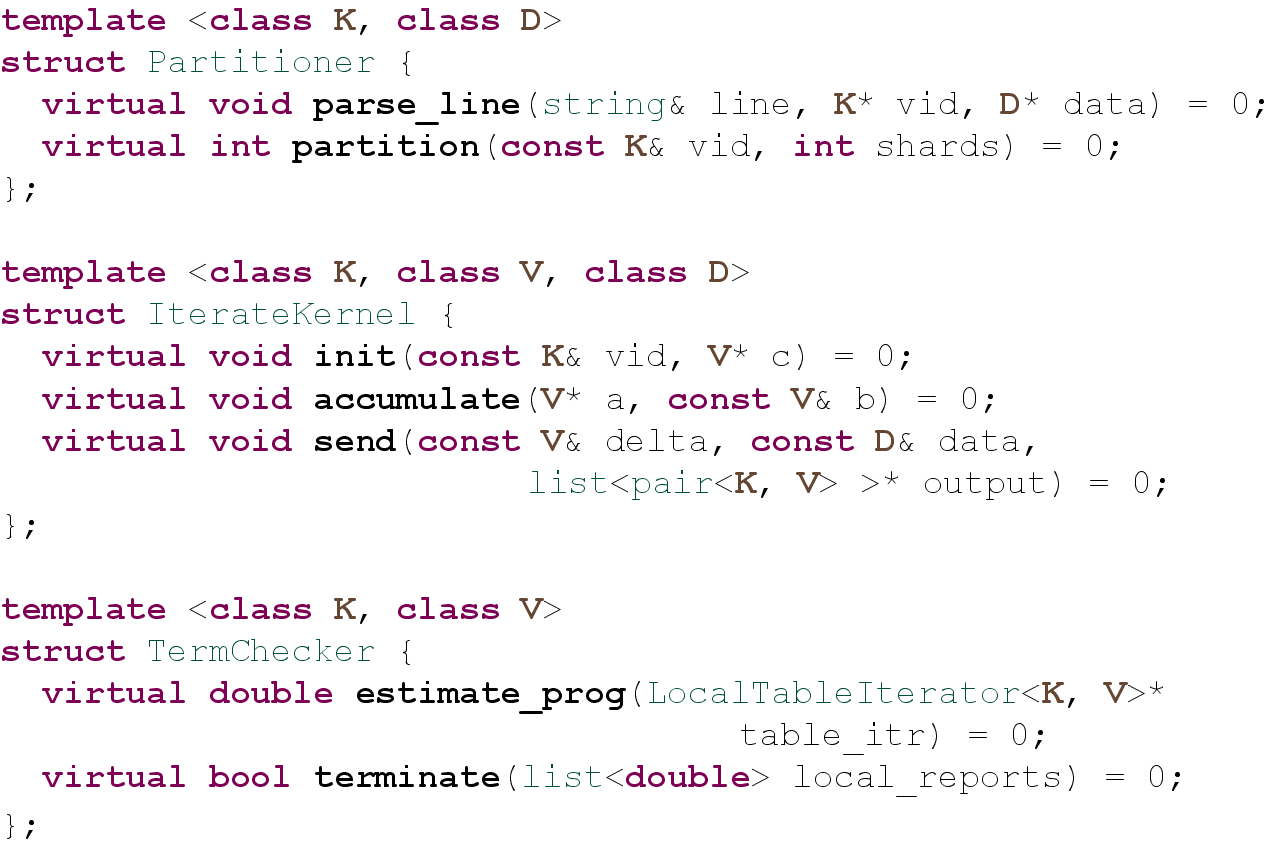}
  \caption{Maiter API summary.}
  \label{fig:api}
\end{figure}

\texttt{K}, \texttt{V}, and \texttt{D} are the template types of data element keys, data element values ($v$ and $\Delta v$), and data element-associate data respectively. Particularly, for each entry in the state table, \texttt{K} is the type of the key field, \texttt{V} is the type of the $v$ field/$\Delta v$ field/priority field, and \texttt{D} is the type of the data field. The \texttt{Partitioner} reads an input partition line by line. The \texttt{parse\_line} function extracts data element id and the associate data by parsing the given line string. Then the \texttt{partition} function applied on the key (e.g., a MOD operation on integer key) determines the host worker of the data element (considering the number of workers/shards). Based on this function, the framework will assign each data element to a host worker and determines a message's destination worker. In the \texttt{IterateKernel} component, users describe a DAIC algorithm by specifying a tuple ($g_{\{i,j\}}(x)$, $\oplus$, $v_j^0$, $\Delta v_j^1$). We initialize $v_j^0$ and $\Delta v_j^1$ by implementing the \texttt{init} interface; specify the `$\oplus$' operation by implementing the \texttt{accumulate} interface; and specify the function $g_{\{i,j\}}(x)$ by implementing the \texttt{send} interface with the given $\Delta v_i$ and data element $i$'s associate data, which generates the output pairs $\langle j, g_{\{i,j\}}(\Delta v_i)\rangle$ to data element $i$'s out-neighbors. To stop an iterative computation, users specify the \texttt{TermChecker} component. The local iteration progress is estimated by specifying the \texttt{estimate\_prog} interface given the local state table iterator. The global terminator collects these local progress reports. In terms of these local progress reports, users specify the \texttt{terminate} interface to decide whether to terminate.

For better understanding, we walk through how the PageRank algorithm is implemented in Maiter \footnote{More implementation example codes are provided at Maiter's Google Code website https://code.google.com/p/maiter/.}. Suppose the input graph file of PageRank is line by line. Each line includes a node id and its adjacency list. The input graph file is split into multiple slices. Each slice is assigned to a Maiter worker. In order to implement PageRank application in Maiter, users should implement three functionality components, \texttt{PRPartitioner}, \texttt{PRIterateKernel}, and \texttt{PRTermChecker}.

\begin{figure}[!t]
  \includegraphics[width=2.4in]{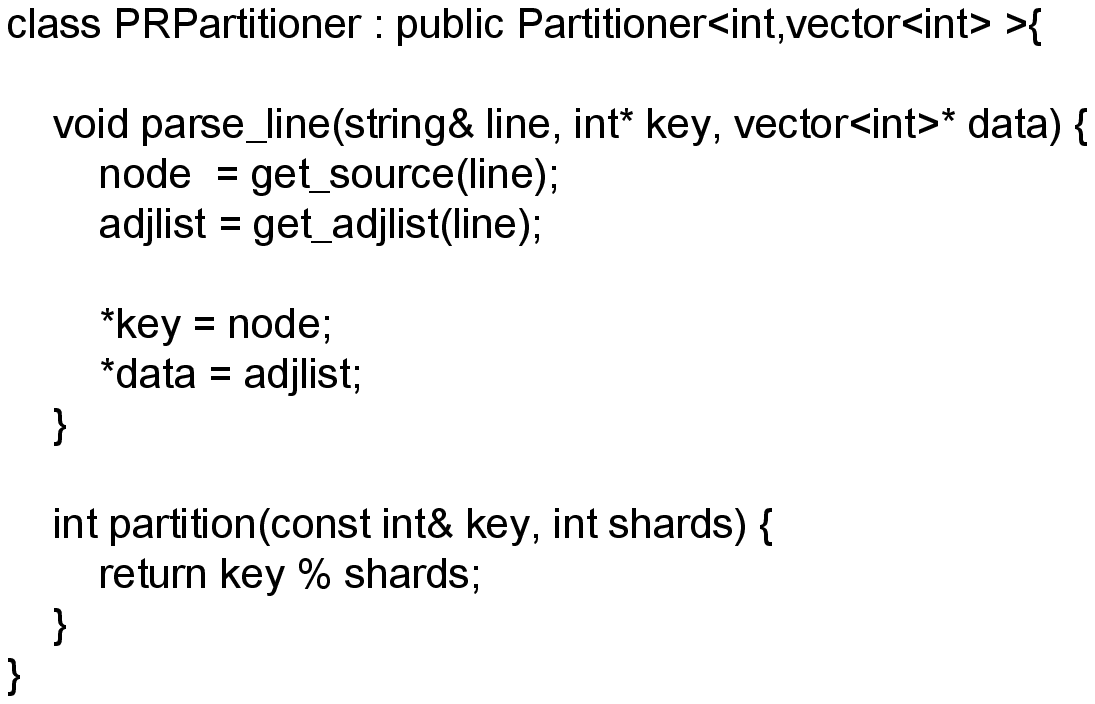}
  \caption{PageRankPartitioner implementation.}
  \label{fig:prpart}
\end{figure}

In \texttt{PRPartitioner}, users specify the \texttt{parse\_line} interface and the \texttt{partition} interface. The implementation code is shown in Fig. \ref{fig:prpart}. In \texttt{parse\_line}, users parse an input line to extract the node id as well as its adjacency list and use them to initialize the state table's key field (\texttt{key}) and data field (\texttt{data}). In \texttt{partition}, users specify the partition function by a simple \emph{mod} operation applied on the key field (\texttt{key}) and the total number of workers (\texttt{shards}).

\begin{figure}[!t]
  \includegraphics[width=3.2in]{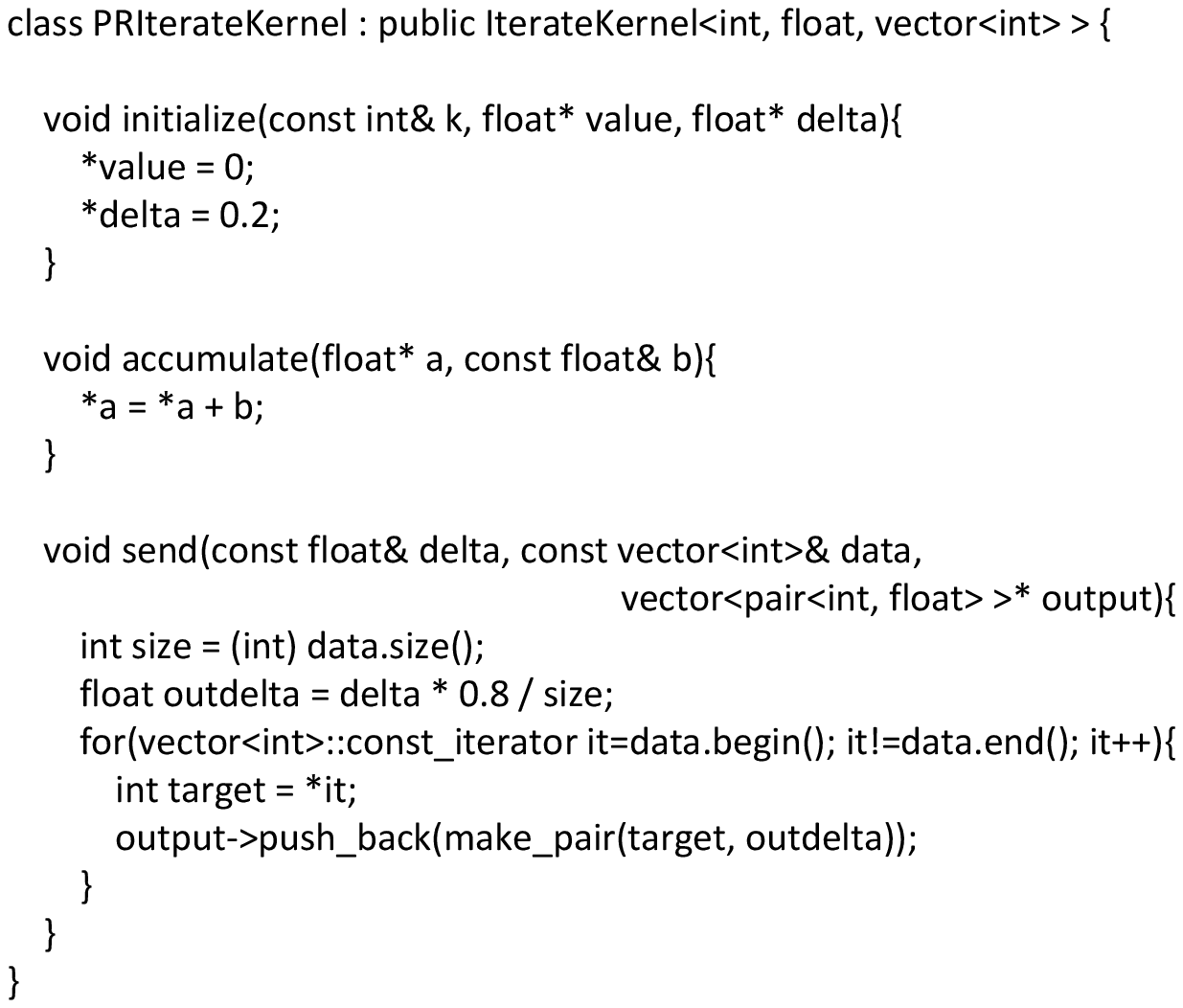}
  \caption{PRIterateKernel implementation.}
  \label{fig:priterate}
\end{figure}

In \texttt{PRIterateKernel}, users specify the asynchronous DAIC process by implementing the \texttt{init} interface, the \texttt{accumulate} interface, and the \texttt{send} interface. The implementation code is shown in Fig. \ref{fig:priterate}. In \texttt{init}, users initialize node $k$'s $v$ field (\texttt{value}) as 0 and $\Delta v$ field (\texttt{delta}) as 0.2. Users specify the \texttt{accumulate} interface by implementing the `$\oplus$' operator as `$+$' (\emph{i.e.}, $a=a+b$). The \texttt{send} operation is invoked after each update of a node. In \texttt{send}, users generate the output messages (contained in \texttt{output}) based on the node's $\Delta v$ value (\texttt{delta}) and data value (\texttt{data}).

\begin{figure}[!t]
  \includegraphics[width=3.0in]{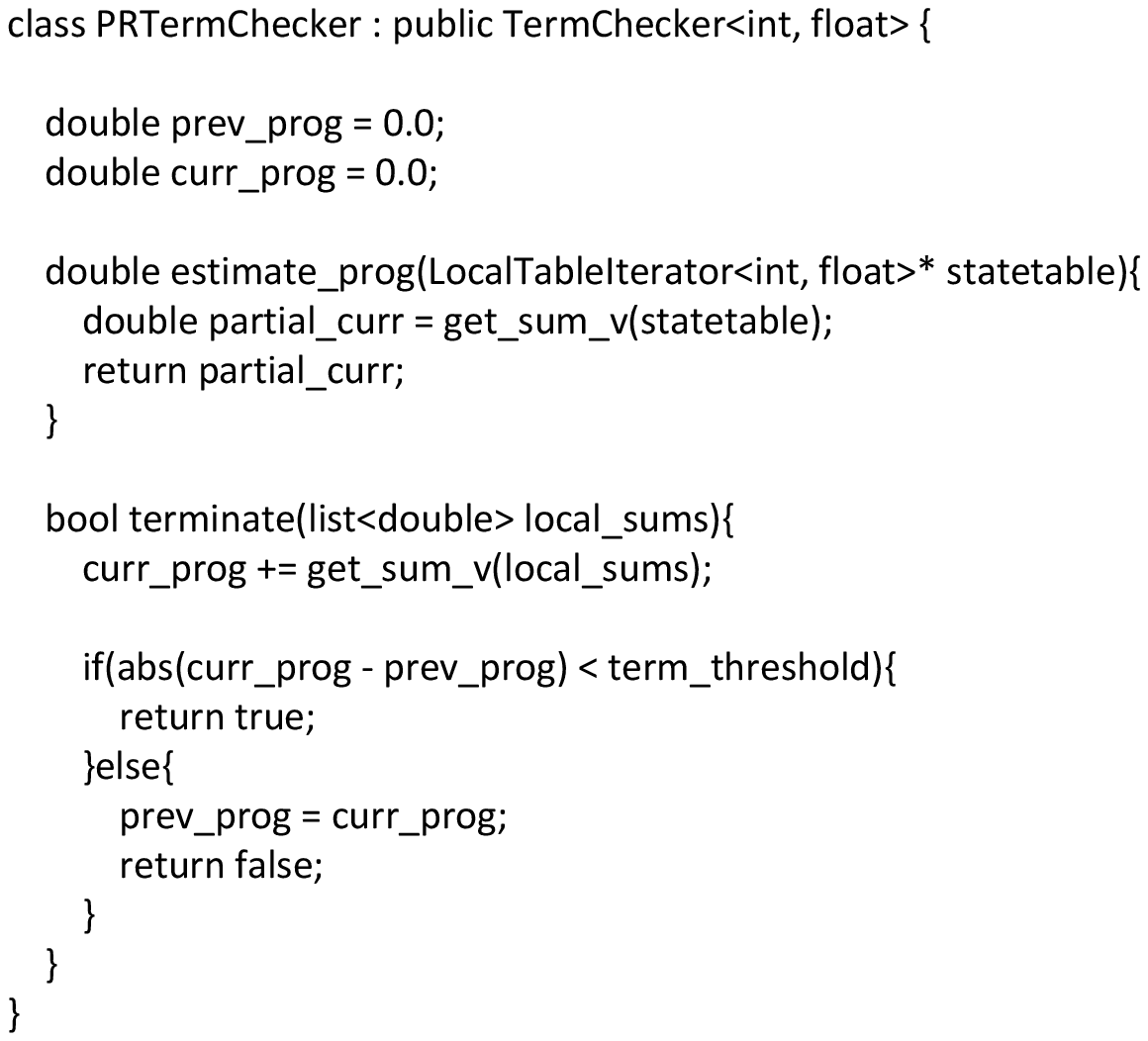}
  \caption{PRTermChecker implementation}
  \label{fig:prterm}
\end{figure}

In \texttt{PRTermChecker}, users specify the \emph{estimate\_prog} interface and the \texttt{terminate} interface. The implementation code is shown in Fig. \ref{fig:prterm}. In \texttt{estimate\_prog}, users compute the summation of $v$ value in local state table. The \emph{estimate\_prog} function is invoked after each period of time. The resulted local sums from various workers are sent to the global termination checker, and then the \texttt{terminate} operation in the global termination checker is invoked. In \texttt{terminate}, based on these received local sums, users compute a global sum, which is considered as the iteration progress. It is compared with the previous iteration's progress to calculate a progress difference. The asynchronous DAIC is terminated when the progress difference is smaller than a pre-defined threshold.

\section{Evaluation}
\label{sec:expr}

This section evaluates Maiter with a series of experiments.

\subsection{Preparation}

The experiments are performed on a cluster of local machines as well as on Amazon EC2 Cloud \cite{amazonec2}. The \emph{local cluster} consisting of 4 commodity machines is used to run small-scale experiments. Each machine has Intel E8200 dual-core 2.66GHz CPU, 3GB of RAM, and 160GB storage. The Amazon EC2 cluster involves 100 medium instances, each with 1.7GB memory and 5 EC2 compute units.

\subsubsection{Frameworks For Comparison}
\label{sec:expr:frameworks}

\textbf{\quad Hadoop} \cite{hadoop} is an open-source MapReduce implementation. It relies on HDFS for storage. Multiple map tasks process the distributed input files concurrently in the map phase, followed by that multiple reduce tasks process the map output in the reduce phase. Users are required to submit a series of jobs to process the data iteratively. The next job operates on the previous job's output. Therefore, two synchronization barriers exist in each iteration, between map phase and reduce phase and between Hadoop jobs. In our experiments, we use Hadoop 1.0.2.

\textbf{iMapReduce} \cite{imapreduce} is built on top of Hadoop and provides iterative processing support. In iMapReduce (\textbf{iMR}), reduce output is directly passed to map rather than dumped to HDFS. More importantly, the iteration variant state data are separated from the static data. Only the state data are processed iteratively, where the costly and unnecessary static data shuffling is eliminated. The original iMapReduce stores data relying on HDFS. iMapReduce can load all data into memory for efficient data access and can store the intermediate data in files for better scalability. We refer to the memory-based iMapReduce as \emph{iMR-mem} and the file-based iMapReduce as \emph{iMR-file}.

\textbf{Spark} \cite{Zaharia:2010:SCC:1863103.1863113} was developed to optimize large-scale iterative and interactive computation. It uses caching techniques and operates in-memory read-only objects to improve the performance for repeated operations. The main abstraction in Spark is resilient distributed dataset (RDD), which maintains several copies of data across memory of multiple machines to support iterative algorithm recovery from failures. The read and write of RDDs is coarse-grained (\emph{i.e.}, read or write a whole block of RDD), so the update of RDDs in iterative computation is coarse-grained. Besides, in Spark, the iteration variant state data can also be separated from the static data by specifying \texttt{partitionBy} and \texttt{join} interfaces. The applications in Spark can be written with Java or Scala. Spark is open-source and can be freely downloaded. In our experiments, we use Spark 0.6.2.

\textbf{PrIter} \cite{priter} enables prioritized iteration by modifying iMapReduce. It exploits the dominant property of some portion of the data and schedules them first for computation, rather than blindly performs computations on all data. The computation workload is dramatically reduced, and as a result the iteration converges faster. However, it performs the priority scheduling in each iteration in a synchronous manner. PrIter provides in-memory version (PrIter 0.1) as well as in-file version (PrIter 0.2). We refer to the memory-based PrIter as \emph{PrIter-mem} and the file-based PrIter as \emph{PrIter-file}.

\textbf{Piccolo} \cite{piccolo} is implemented with C++ and MPI, which allows to operate distributed tables. The iterative algorithm can be implemented by updating the distributed tables iteratively. The intermediate data are shuffled between workers continuously as long as some amount of the intermediate data are produced (fine-grained write), instead of waiting for the end of iteration and sending them together. The current iteration's data and the next iteration's data are stored in two global tables separately, so that the current iteration's data will not be overwritten. Piccolo can maintain the global table both in memory and in file. We only consider the in-memory version.

\textbf{GraphLab} ~\cite{DBLP:journals/corr/abs-1006-4990} supports both synchronous and asynchronous iterative computation with sparse computational dependencies while ensuring data consistency and achieving a high degree of parallel performance. It is also implemented with C++ and MPI. It first only supports the computation under multi-core environment exploiting shared memory (GraphLab 1.0). But later, GraphLab supports large-scale distributed computation under cloud environment (GraphLab 2.0). The static data and dynamic data in GraphLab can be decoupled and the update of vertex/edge state in GraphLab is fine-grained. Under asynchronous execution, several scheduling policies including FIFO scheduling and priority scheduling are supported in Graphlab. GraphLab performs a fine-grained termination check. It terminates a vertex's computation when the change of the vertex state is smaller than a pre-defined threshold parameter.

\begin{table}[!t]
    \caption{Comparison of Distributed Frameworks}
    \label{tab:frameworks}
    \centering
    \begin{tabular}{c|c|c|c|c|c}
    \hline
    \multirow{2}{*}{\bfseries name} & sep & in & fine-g & async & pri\\
                                    & data & mem & update & iter & sched\\
\hline\hline
 Hadoop & \texttimes & \texttimes & \texttimes & \texttimes & \texttimes \\
\hline
 iMR-file & $\checkmark$ & \texttimes & \texttimes & \texttimes & \texttimes \\
\hline
 iMR-mem & $\checkmark$ & $\checkmark$ & \texttimes & \texttimes & \texttimes \\
\hline
 Spark & $\checkmark$  & $\checkmark$ & \texttimes & \texttimes & \texttimes\\
\hline
 PrIter-file & $\checkmark$ & \texttimes & \texttimes & \texttimes & $\checkmark$ \\
\hline
 PrIter-mem & $\checkmark$ & $\checkmark$ & \texttimes & \texttimes & $\checkmark$ \\
\hline
 Piccolo & $\checkmark$  & $\checkmark$ & $\checkmark$ & \texttimes & \texttimes\\
\hline
 GraphLab-Sync & $\checkmark$  & $\checkmark$ & $\checkmark$ & \texttimes & \texttimes\\
\hline
 GraphLab-AS-fifo & $\checkmark$ & $\checkmark$ & $\checkmark$ & $\checkmark$ & \texttimes\\
\hline
 GraphLab-AS-pri & $\checkmark$ & $\checkmark$ & $\checkmark$ & $\checkmark$ & $\checkmark$\\
\hline
 Maiter-Sync & $\checkmark$ & $\checkmark$ & $\checkmark$ & \texttimes & \texttimes\\
\hline
 Maiter-RR & $\checkmark$ & $\checkmark$ & $\checkmark$ & $\checkmark$ & \texttimes\\
\hline
 Maiter-Pri & $\checkmark$ & $\checkmark$ & $\checkmark$ & $\checkmark$ & $\checkmark$\\
\hline
\end{tabular}
\end{table}

Table \ref{tab:frameworks} summarizes these frameworks. These frameworks are featured by various factors that help improve performance, including separating static data from state data (sep data), in-memory operation (in mem), fine-grained update (fine-g update), asynchronous iteration (async iter), and the priority scheduling mechanism under asynchronous iteration engine (pri sched).

\subsubsection{Applications and Data Sets}
Four applications, including PageRank, SSSP, Adsorption, and Katz metric, are implemented. We use Google Webgraph \cite{datasets} for PageRank computation. 

We generate synthetic massive data sets for these algorithms. The graphs used for SSSP and Adsorption are weighted, and the graphs for PageRank and Katz metric are unweighted. The node ids are continuous integers ranging from 1 to size of the graph. We decide the in-degree of each node following log-normal distribution, where the log-normal parameters are ($\mu=-0.5$, $\sigma=2.3$). Based on the in-degree of each node, we randomly pick a number of nodes to point to that node. For the weighted graph of SSSP computation, we use the log-normal parameters ($\mu=0$, $\sigma=1.0$) to generate the float weight of each edge following log-normal distribution. For the weighted graph of Adsorption computation, we use the log-normal parameters ($\mu=0.4$, $\sigma=0.8$) to generate the float weight of each edge following log-normal distribution. These log-normal parameters for these graphs are extracted from a few small real graphs downloaded from \cite{datasets}.

\subsubsection{Termination Condition of the Experiments}

To terminate iterative computation in PageRank experiment, we first run PageRank off-line to obtain a resulted rank vector, which is assumed to be the converged vector $R^*$. Then we run PageRank with different frameworks. We terminate the PageRank computation when the L1-Norm distance between the iterated vector $R$ and the converged vector $R^*$ is less than $0.001\cdot N$, where $N$ is the total number of nodes, \emph{i.e.}, $\sum_{j}(|R^*_j-R_j|)<0.001\cdot N$. For the synchronous frameworks (\emph{i.e.}, Hadoop, iMR-file, iMR-mem, Spark, PrIter-file, PrIter-mem, Piccolo, and Maiter-Sync), we check the convergence (termination condition) after every iteration. For the asynchronous frameworks (\emph{i.e.}, Maiter-RR, and Maiter-Pri), we check the convergence every termination check interval. For GraphLab variants, we set the parameter of convergence tolerance as 0.001 to terminate the computation. Note that, the time for termination check in Hadoop and Piccolo (computing the L1-Norm distance through another job) has been excluded from the total running time, while the other frameworks provide termination check functionality.

For SSSP, the computation is terminated when there is no update of any vertex. For Adsorption and Katz metric, we use the similar convergence check approach as PageRank.

\subsection{Running Time to Convergence}

\textbf{Local Cluster Results.} We compare different frameworks on running time in the context of PageRank computation. Due to the limited space, the termination approach of PageRank computation is presented in Section 5.3 of the supplementary file. Fig. \ref{fig:platform:local} shows the PageRank running time on Google Webgraph on our local cluster. Note that, the data loading time for the memory-based systems (other than Hadoop, iMR-file, iMR-mem) is included in the total running time.

\begin{figure}[!t]
  \centering
  \includegraphics[width=2.5in]{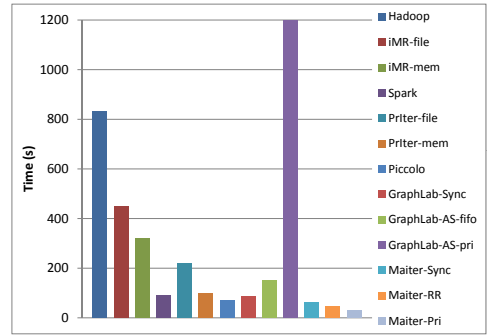}
  \caption{Running time of PageRank on Google Webgraph on local cluster.}
  \label{fig:platform:local}
  \vspace{-3mm}
\end{figure}

By using Hadoop, we need 27 iterations and more than 800 seconds to converge. By separating the iteration-variant state data from the static data, iMR-file reduces the running time of Hadoop by around 50\%. iMR-mem further reduces it by providing faster memory access. Spark, with efficient data partition and memory caching techniques, can reduce Hadoop time to less than 100 seconds. PrIter identifies the more important nodes to perform the update and ignores the useless updates, by which the running time is reduced. As expected, PrIter-mem converges faster than PrIter-file. Piccolo utilizes MPI for message passing to realize fine-grained updates, which improves the performance.

GraphLab variants show their differences on the performance. GraphLab-Sync uses a synchronous engine and completes the iterative computation within less than 100 seconds. GraphLab-AS-fifo uses an asynchronous engine and schedules the asynchronous updates in a FIFO queue, which consumes much more time. The reason is that the cost of managing the scheduler (through locks) tends to exceed the cost of the main PageRank computation itself. The cost of maintaining the priority queue under asynchronous engine seems even much larger, so that GraphLab-AS-pri converges with significant longer running time. More experimental results focusing on demonstrating the difference between GraphLab and Maiter can be found in Section 5.4 of the supplementary file.

The framework that supports synchronous DAIC, Maiter-Sync, filters the zero updates ($\Delta R=0$) and reduces the running time to about 60 seconds. Further, the asynchronous DAIC frameworks, Maiter-RR and Maiter-Pri, can even converge faster by avoiding the synchronous barriers. Note that, our priority scheduling mechanism does not result in high cost, since we do not need distributed lock for scheduling asynchronous DAIC. In addition, in priority scheduling, the approximate sampling technique \cite{priter} helps reduce the complexity, which avoids high scheduling cost.

\begin{figure}[!t]
  \centering
  \includegraphics[width=2.2in]{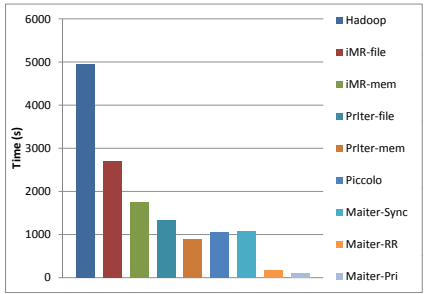}
  \caption{Running time of PageRank on 100-million-node synthetic graph on EC2 cluster.}
  \label{fig:platform:ec2}
  \vspace{-4mm}
\end{figure}

\textbf{EC2 Results.} To show the performance under large-scale distributed environment, we run PageRank on a 100-million-node synthetic graph on EC2 cluster. Fig. \ref{fig:platform:ec2} shows the running time with various frameworks. We can see the similar results. One thing that should be noticed is that Maiter-Sync has comparable performance with Piccolo and PrIter. Only DAIC is not enough to make a significant performance improvement. However, the asynchronous DAIC frameworks (Maiter-RR and Maiter-Pri) perform much better. The result is under expectation. As the cluster size increases and the heterogeneity in cloud environment becomes apparent, the problem of synchronous barriers is more serious. With the asynchronous execution engine, Maiter-RR and Maiter-Pri can bypass the high-cost synchronous barriers and perform more efficient computations. As a result, Maiter-RR and Maiter-Pri significantly reduce the running time. Moreover, Maiter-Pri exploits more efficient priority scheduling, which can achieve 60x speedup over Hadoop. This result demonstrates that only with asynchronous execution can DAIC reach its full potential.

\begin{figure}[!t]
  \centering
  \includegraphics[width=2.3in]{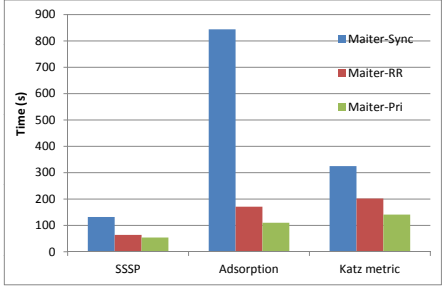}
  \caption{Running time of other applications (SSSP, Adsorption, and Katz metric) on EC2 cluster.}
  \label{fig:platform:otherapp}
  \vspace{-3mm}
\end{figure}

To show that Maiter can support more applications, we also run other applications on EC2 cluster. We perform SSSP, Adsorption, and Katz metric computations with Maiter-Sync, Maiter-RR, and Maiter-Pri. We generate weighted/unweighted 100-million-node synthetic graphs for these applications respectively. Fig. \ref{fig:platform:otherapp} shows the running time of these applications. For SSSP, the asynchronous DAIC SSSP (Maiter-RR and Maiter-Pri) reduces the running time of synchronous DAIC SSSP (Maiter-Sync) by half. For Adsorption, the asynchronous DAIC Adsorption is 5x faster than the synchronous DAIC Adsorption. Further, by priority scheduling, Maiter-Pri further reduces the running time of Maiter-RR by around 1/3. For Katz metric, we can see that Maiter-RR and Maiter-Pri also outperform Maiter-Sync.

\subsection{Efficiency of Asynchronous DAIC}
\label{sec:expr:updates}

As analyzed in Section \ref{sec:analysis}, with the same number of updates, asynchronous DAIC results in more progress than synchronous DAIC. In this experiment, we measure the number of updates that PageRank and SSSP need to converge under Maiter-Sync, Maiter-RR, and Maiter-Pri. In order to measure the iteration process, we define a \emph{progress metric}, which is $\sum_{j}R_j$ for PageRank and $\sum_{j}d_j$ for SSSP. Then, the \emph{efficiency} of the update operations can be seen as the ratio of the progress metric to the number of updates.

\begin{figure}[!htb]
\vspace{-4mm}
    \centerline{\subfloat[PageRank]{\includegraphics[width=1.25in, angle=-90]{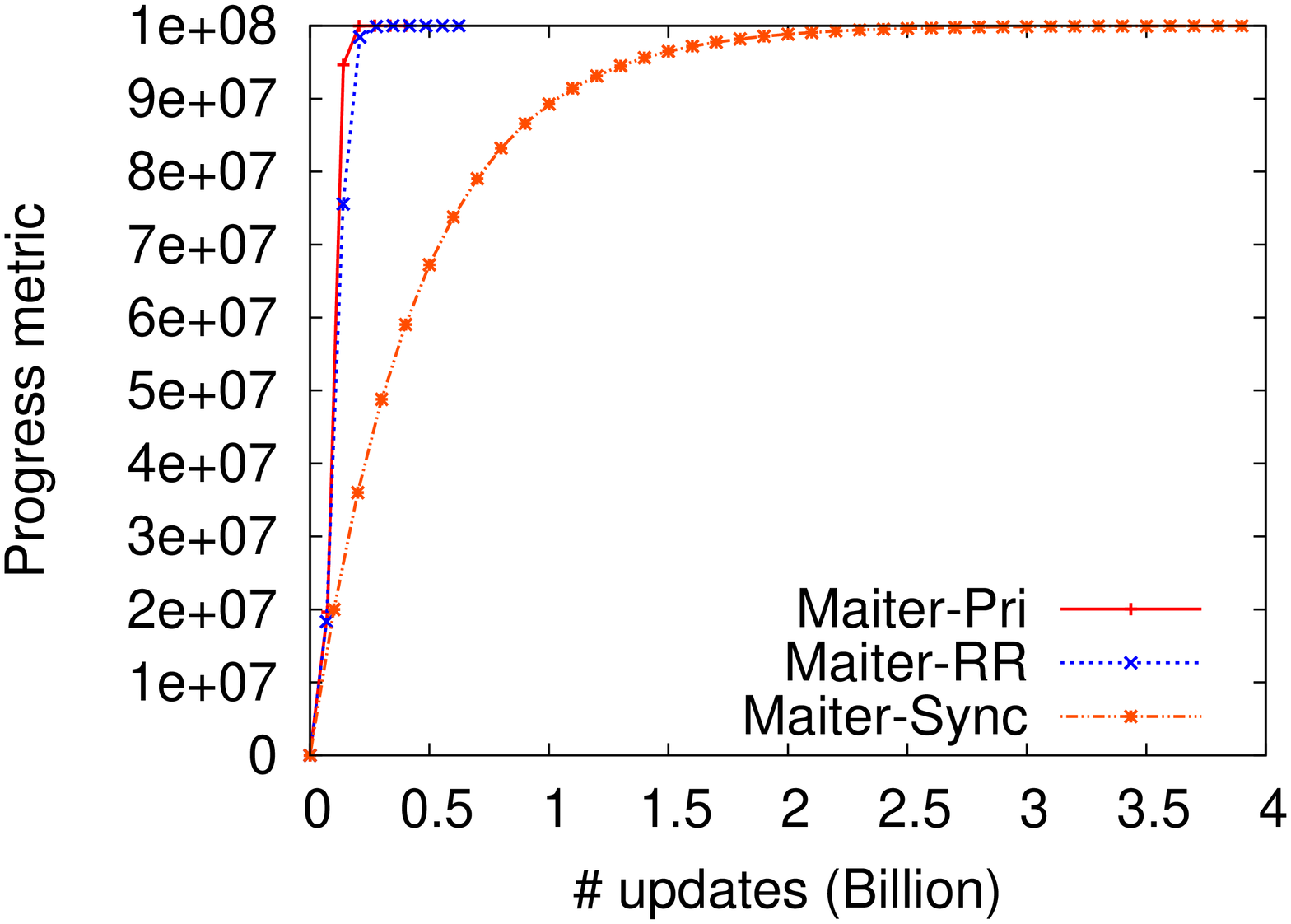}
    \label{fig:updates:pr}}
    \hfill
    \subfloat[SSSP]{\includegraphics[width=1.25in, angle=-90]{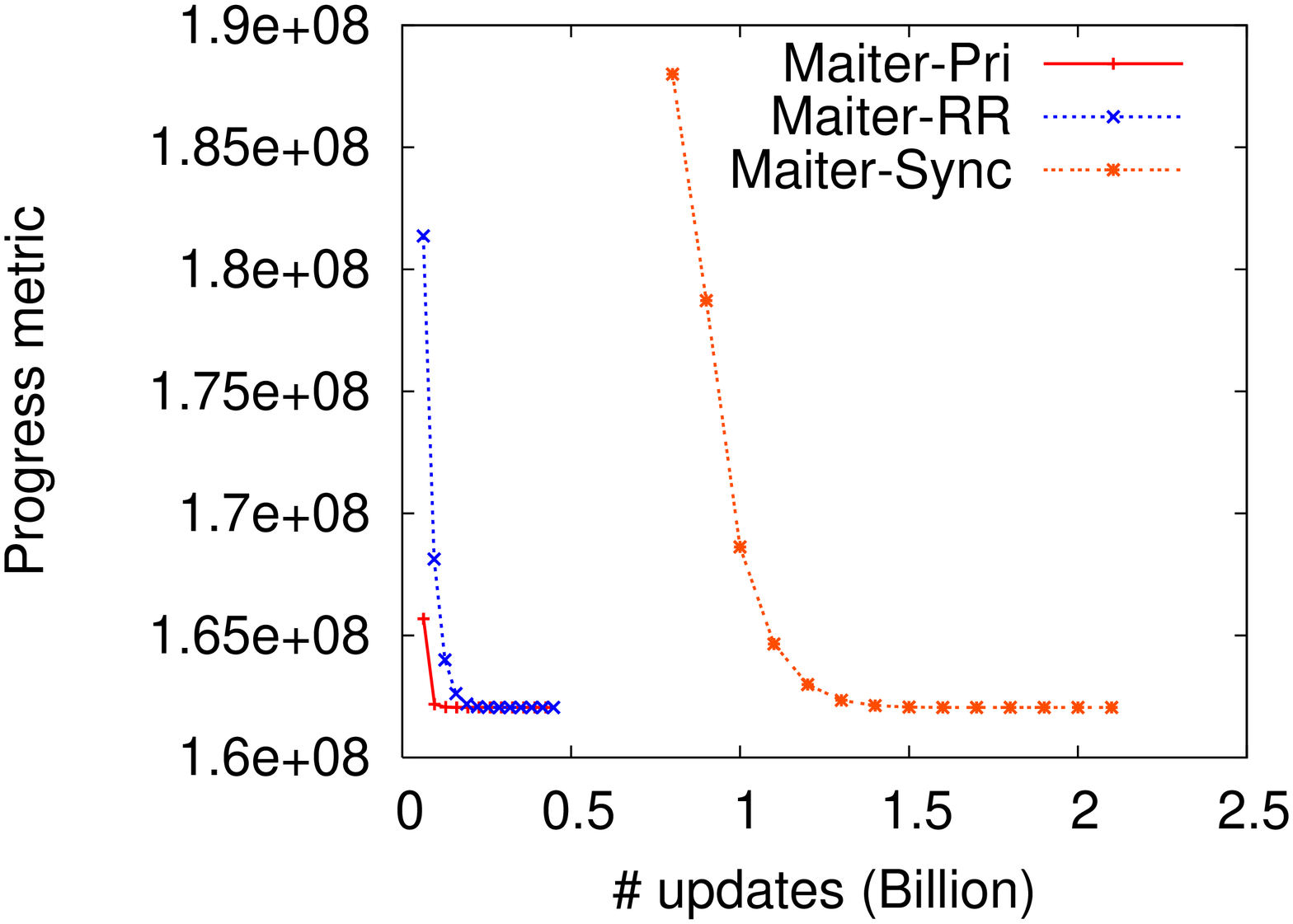}
    \label{fig:updates:sp}}
    }
    \caption{Number of updates vs. progress metric.}
    \label{fig:updates}
\end{figure}

On the EC2 cluster, we run PageRank on a 100-million-node synthetic graph and SSSP on a 500-million-node synthetic graph. Fig. \ref{fig:updates:pr} shows the progress metric against the number of updates for PageRank. In PageRank, the progress metric $\sum_{j}R_j$ should be increasing. Each $R_j^0$ is intialized to be 0 and each $\Delta R_j^1$ is initialized to be $1-d=0.2$ (the damping factor $d=0.8$). The progress metric $\sum_{j}R_j$ is increasing from $\sum_{j}R_j^1=\sum_{j}(R_j^0+\Delta R_j^1)=0.2\cdot N$ to $N$, where $N=10^{8}$ (number of nodes). Fig. \ref{fig:updates:sp} shows the progress metric against the number of updates for SSSP. In SSSP, the progress metric $\sum_{j}d_j$ should be decreasing. Since $d_j$ is initialized to be $\infty$ for any node $j\neq s$, which cannot be drawn in the figure, we start plotting when any $d_j<\infty$. From Fig. \ref{fig:updates:pr} and Fig. \ref{fig:updates:sp}, we can see that by asynchronous DAIC, Maiter-RR and Maiter-Pri require much less updates to converge than Maiter-Sync. That is, the update in asynchronous DAIC is more effective than that in synchronous DAIC. Further, Maiter-Pri selects more effective updates to perform, so the update in Maiter-Pri is even more effective.

\subsection{Scaling Performance}

Suppose that the running time on one worker is $T$. With optimal scaling performance, the running time on an $n$-worker cluster should be $\frac{T}{n}$. But in reality, distributed application usually cannot achieve the optimal scaling performance. In order to measure how asynchronous Maiter scales with increasing cluster size, we perform PageRank on a 100-million-node graph on EC2 as the number of workers increases from 20 to 100. We consider the running time on a 20-worker cluster as the baseline, based on which we determine the running time with optimal scaling performance on different size clusters. We consider Hadoop, Maiter-Sync, Maiter-RR, and Maiter-Pri for comparing their scaling performance.

\begin{figure}[!t]
  \centering
  \includegraphics[width=1.4in, angle=-90]{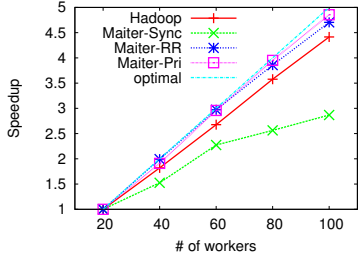}
  \caption{Scaling performance.}
  \label{fig:scale}
  \vspace{-2mm}
\end{figure}

Fig. \ref{fig:scale} shows the scaling performance of Hadoop, Maiter-Sync, Maiter-RR, and Maiter-Pri. We can see that the asynchronous DAIC frameworks, Maiter-RR and Maiter-Pri, provide near-optimal scaling performance as cluster size scales from 20 to 100. The performance of the synchronous DAIC framework Maiter-Sync is degraded a lot as the cluster size scales. Hadoop splits a job into many fine-grained tasks (task with 64MB block size), which alleviates the impact of synchronization and improves scaling performance.

\begin{figure}[!htb]
  \centering
  \includegraphics[width=1.5in, angle=-90]{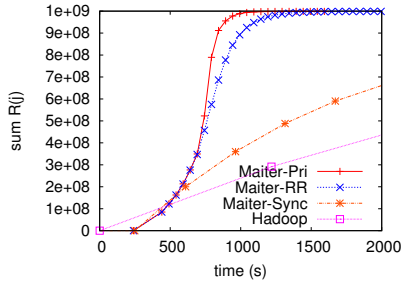}
  \caption{Running time vs. progress metric of PageRank on a 1-billion-node synthetic graph.}
  \label{fig:pr1G}
\end{figure}

In order to measure how Maiter scales with increasing input size, we perform PageRank for a 1-billion-node graph on the 100-node EC2 cluster. Maiter runs normally without any problem. Figure \ref{fig:pr1G} shows the progress metric against the running time of Hadoop, Maiter-Sync, Maiter-RR, and Maiter-Pri. Since it will take considerable long time for PageRank convergence in Hadoop and Maiter-Sync, we only plot the progress changes in the first 2000 seconds. Maiter-Sync, Maiter-RR, and Maiter-Pri spend around 240 seconds in loading data in memory before starting computation. The PageRank computations in the asynchronous frameworks (Maiter-RR and Maiter-Pri) converge much faster than that in the synchronous frameworks (Hadoop and Maiter-Sync). In addition, to evaluate how large graph Maiter can process at most in the 100-node EC2 cluster, we continue to increase the graph size to contain 2 billion nodes, and it works fine with memory usage up to 84.7\% on each EC2 instance.

\subsection{Comparison of Asynchronous Frameworks: Maiter vs. GraphLab}

In this experiment, we focus on comparing Maiter with another asynchronous framework GraphLab. Even though GraphLab support asynchronous computation, as shown in Fig. 2 of the TPDS manuscript, it shows poor performance under asynchronous execution engine. Especially for priority scheduling, it extremely extends the completion time.

GraphLab relies on chromatic engine (partially asynchronously) and distributed locking engine (fully asynchronous) for scheduling asynchronous computation. Distributed locking engine is costly, even though many optimization techniques are exploited in GraphLab. For generality, the scheduling of asynchronous computation should guarantee the dependencies between computations. Distributed locking engine is proposed for the generality, but it becomes the bottleneck of asynchronous computation. Especially for priority scheduling, the cost of managing the scheduler tends to exceed the cost of the PageRank computation itself, which leads to very slow asynchronous Pagerank computation in GraphLab. Actually, GraphLab's priority scheduling policy is designed for some high-workload applications, such as Loopy Belief Propagation \cite{Ihler:2005:LBP:1046920.1088703}, in which case the asynchronous computation advantage is much more substantial.

\begin{figure}[!htb]
  \centering
  \includegraphics[width=2.3in]{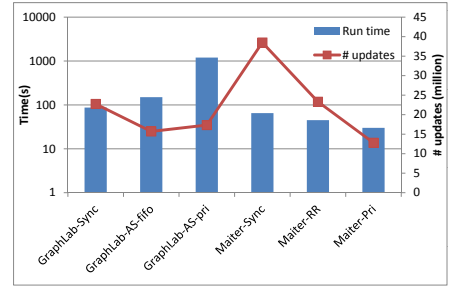}
  \caption{Running time and number of updates of PageRank computation on GraphLab and Maiter.}
  \label{fig:vsgl}
\end{figure}

To verify our analysis, we run PageRank on Maiter and GraphLab to compare the running time and the number of updates. The experiment is launched in the local cluster, and the graph dataset is the Google Webgraph dataset. Fig. \ref{fig:vsgl} shows the result. In GraphLab, the number of performed updates under asynchronous engine (both fifo scheduling and priority scheduling) is less than that under synchronous engine, but the running time is longer. Under asynchronous engine, the number of updates by priority scheduling is similar to that by fifo scheduling, but the running time is extremely longer. Even though the workload is reduced, the asynchronous scheduling becomes an extraordinarily costly job, which slows down the whole process.

On the contrary, asynchronous DAIC exploits the cumulative operator `$\oplus$', which has commutative property and associative property. This implicates that the delta values can be accumulated in any order and at any time. Therefore, Maiter does not need to guarantee the computation dependency while allows all vertices to update their state totally independently. Round-robin scheduling, which performs computation on the local vertices in a round-robin manner, is the easiest one to implement (\emph{i.e.}, with low overhead). Further, priority scheduling identifies the vertex importance and executes computation in their importance order, which can accelerate convergence. Both of them do not need to guarantee the global consistency and do not result in serious overhead. As shown in Fig. \ref{fig:vsgl}, round-robin scheduling and priority scheduling first reduce the workload (less number of updates), and as result shorten the convergence time.

\subsection{Communication Cost}
\label{sec:commovhd}

Distributed applications need high-volume communication between workers. The communication between workers becomes the performance bottleneck. Saving the communication cost correspondingly helps improve performance. By asynchronous DAIC, the iteration converges with much less number of updates, and as a result needs less communication.

\begin{figure}[!htb]
  \centering
  \includegraphics[width=2.6in]{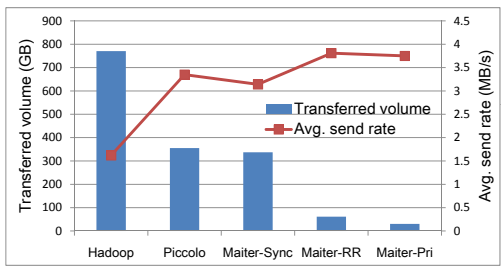}
  \caption{Communication cost.}
  \label{fig:commcost}
\end{figure}

To measure the communication cost, we run PageRank on a 100-million-node synthetic graph on the EC2 cluster. We record the amount of data sent by each worker and sum these amounts of all workers to obtain the total volume of data transferred. Figure \ref{fig:commcost} depicts the total volume of data transferred in Hadoop, Piccolo, Maiter-Sync, Maiter-RR, and Maiter-Pri. We choose Hadoop for comparison for its generality and popularity. Hadoop mixes the iteration-variant state data with the static data and shuffles them in each iteration, which results in high volume communication. Piccolo can separate the state data from the static data and only communicate the state data. Besides, unlike the file-based transfer in Hadoop, Piccolo communicates between workers through MPI. As shown in the figure, Piccolo results in less transferred volume than Hadoop. Maiter-Sync utilizes msg tables for early aggregation to reduce the total transferred volume in a certain degree. By asynchronous DAIC, we need less number of updates and as a result less amount of communication. Consequently, Maiter-RR and Maiter-Pri significantly reduce the transferred data volume. Further, Maiter-Pri transfers even less amount of data than Maiter-RR since Maiter-Pri converges with even less number of updates. Maiter-RR and Maiter-Pri run significantly faster, and at the same time the amount of shuffled data is much less.

In Figure \ref{fig:commcost}, we also show the average bandwidth that each worker has used for sending data. The worker in Maiter-RR and Maiter-Pri consumes about 2 times bandwidth than that in Hadoop and consumes only about 20\% more bandwidth than the synchronous frameworks, Piccolo and Maiter-Sync. The average consumed bandwidth in asynchronous DAIC frameworks is a little higher. This means that the bandwidth resource in a cluster is highly utilized.

\section{Related Work}
\label{sec:relatedwork}

The original idea of asynchronous iteration, chaotic iteration, was introduced by Chazan and Miranker in 1969 \cite{Chazan1969199}. Motivated by that, Baudet proposed an asynchronous iterative scheme for multicore systems \cite{Baudet:1978:AIM:322063.322067}, and Bertsekas presented a distributed asynchronous iteration model \cite{Bertsekas83distributedasynchronous}. These early stage studies laid the foundation of asynchronous iteration and have proved its effectiveness and convergence. Asynchronous methods are being increasingly used and studied since then, particularly so in connection with the use of heterogeneous workstation clusters. A broad class of applications with asynchronous iterations have been correspondingly raised \cite{Frommer:2000:AI:363882.363901,Miellou:1998:NCA:279724.279745}, such as PageRank \cite{McSherry:2005:UAA:1060745.1060829,conf/parco/KolliasGS05} and pairwise clustering \cite{yom-tov:parallel}. Our work differs from these previous works. We focus on a particular class of iterative algorithms and provide a new asynchronous iteration scheme, DAIC, which exploits the accumulative property.

On the other hand, to support iterative computation, a series of distributed frameworks have emerged. In addition to the frameworks we compared in Section~\ref{sec:expr}, many other synchronous frameworks are proposed recently. HaLoop~\cite{haloop}, a modified version of Hadoop, improves the efficiency of iterative computations by making the task scheduler loop-aware and employing caching mechanisms. CIEL~\cite{ciel} supports data-dependent iterative algorithms by building an abstract dynamic task graph. Pregel~\cite{1807184} aims at supporting graph-based iterative algorithms by proposing a graph-centric programming model. REX \cite{rex} optimizes DBMS recursive queries by using incremental updates. Twister~\cite{Ekanayake:2010:TRI:1851476.1851593} employs a lightweight iterative MapReduce runtime system by logically constructing a reduce-to-map loop. Naiad \cite{naiad} is recently proposed to support incremental iterative computations.

All of the above described works build on the basic assumption that the synchronization between iterations is essential. A few proposed frameworks also support asynchronous iteration. The partial asynchronous approach proposed in \cite{5600303} investigates the notion of partial synchronizations in iterative MapReduce applications to overcome global synchronization overheads.GraphLab~\cite{DBLP:journals/corr/abs-1006-4990} supports asynchronous iterative computation with sparse computational dependencies while ensuring data consistency and achieving a high degree of parallel performance. PowerGraph \cite{Gonzalez:2012:PDG:2387880.2387883} forms the foundation of GraphLab, which characterizes the challenges of computation on natural graphs. The authors propose a new approach to distributed graph placement and representation that exploits the structure of power-law graphs. GRACE \cite{grace} executes iterative computation with asynchronous engine while letting users implement their algorithms with the synchronous BSP programming model. To the best of our knowledge, our work is the first that proposes to perform DAIC for iterative algorithms. We also identify a broad class of iterative algorithms that can perform DAIC.

\section{Conclusions}
\label{sec:conclusion}

In this paper, we propose DAIC, delta-based accumulative iterative computation. The DAIC algorithms can be performed asynchronously and converge with much less workload. To support DAIC model, we design and implement Maiter, which is running on top of hundreds of commodity machines and relies on message passing to communicate between distributed machines. We deploy Maiter on local cluster as well as on Amazon EC2 cloud to evaluate its performance in the context of four iterative algorithms. The results show that by asynchronous DAIC the iterative computation performance is significantly improved.

\small
\bibliographystyle{abbrv}
\bibliography{./ref}


%\vfill


\end{document}